\documentclass[journal,10pt]{IEEEtran}
\usepackage{amsthm,amsfonts}
\usepackage{amssymb}
\usepackage{algorithm}
\usepackage{algorithmic}
\usepackage{balance} 
\usepackage{cite,graphicx,amsmath,amssymb}
\usepackage{subfigure}
\usepackage{citesort}
\usepackage{fancyhdr}
\usepackage{tabularx}
\usepackage{mdwmath,mathtools}
\usepackage{mdwtab}
\usepackage{enumitem}
\usepackage{xcolor}
\usepackage{bm}
\usepackage{caption}
\usepackage{multirow}
\usepackage{flafter}
\usepackage{url}
\usepackage{longtable}
\usepackage{subfloat}
\usepackage{makecell}

\newtheorem{proposition}{Proposition}
\newtheorem{lemma}{Lemma}
\newtheorem{remark}{Remark}
\newcommand{\RomanNumeralCaps}[1]{\MakeUppercase{\romannumeral #1}}

\begin{document}
	\title{The Impact of Uniform Circular Array on Near-field ISAC}
 
\author{ 
	Na~Xue,~\IEEEmembership{Graduate Student Member,~IEEE}, Xidong~Mu,~\IEEEmembership{Member,~IEEE},   Yue~Chen,~\IEEEmembership{Senior Member,~IEEE},
	Yuanwei~Liu,~\IEEEmembership{Fellow,~IEEE}
	\thanks{Part of this paper has been accepted at the IEEE International Conference on Communications (ICC) workshop, Montreal, Canada, 8–12 Jun, 2025~\cite{Na_conf4}. \\
N. Xue, Y. Chen are with Queen Mary University of London, London, UK (email:\{n.xue,yue.chen\}@qmul.ac.uk). X. Mu is with Queen's University Belfast, Belfast, UK (email: x.mu@qub.ac.uk). Y. Liu are with the University of Hong Kong, Hong Kong (email: yuanwei@hku.hk). }     }
	\maketitle
	\vspace{-2.5cm}
	
	\begin{abstract} 
        A novel uniform circular array (UCA) based near-field (NF) integrated sensing and communication (ISAC) framework is proposed, where the Cylindrical coordinate is invoked to evaluate the joint positioning performance. The joint squared position error bound (SPEB) of the sensing target (ST) is derived for the coplanar and non-coplanar cases. For the coplanar case, where the ST is located in the coplanar region of the UCA, the approximate Cram{\'e}r-Rao bound (CRB) expressions for the separate angle and distance estimation are given by exploiting the uniform spherical wavefront model. A SPEB minimization problem is formulated with the constraints of communication requirement and power budget, where the closed-form solution to minimize the CRB of the angle is derived. Inspired by the close-form expression, a low complexity vector-based quadratic transformation (VQF) algorithm is proposed by invoking the Rayleigh quotient. For the non-coplanar case, where the ST is located beyond the coplanar region of the UCA, the separate CRBs over three-dimensional coordinates and the joint SPEB approximations are derived. To minimize the SPEB performance, the semi-definite relaxation (SDR) method and extended low-complexity VQF algorithm are proposed. Numerical results validated that i) the Fisher Information Matrix about angle and distance in NF propagation can be approximated as a diagonal matrix with the trinity loss; ii) Compared with the uniform planar array, the UCA achieve better positioning performance when ST located in the coplanar of the antenna array; and iii) the proposed VQF algorithms reach higher solution precision than conventional SDR algorithm with much less computation complexity. 
	\end{abstract}
	
	\begin{IEEEkeywords}
		Integrated sensing and communication (ISAC),  Near-field (NF), antenna design, beamformer design
	\end{IEEEkeywords}
	
    \section{Introduction}
    Wireless sensing functionality has been regarded as an indispensable component in sixth-generation (6G) wireless networks~\cite{Kim_24WC}. On the one hand, emerging technologies, such as unmanned aerial vehicles (UAV), Internet of Vehicles, and extended reality (XR), rely on the perception of the environment to improve user experience~\cite{Saad20Network,Metaverse}. On the other hand, the mentioned technologies exacerbate the complex variations in electromagnetic wavefront (EM) propagation, which imposes a more stringent requirement on ultra-high data rate and high-precision sensing~\cite{Letaief_19CM,XiaodanShan_22JSAC}. Benefiting from the EM propagation and similarity of hardware design among radio sensing functionality and wireless communication functionality, integrated sensing and communication (ISAC) has been deemed as a promising candidate technology in 6G. By unifying these dual functionalities via the shared spectrum and hardware design, ISAC exploits the EM wavefront to convey the modulated information as well as conduct the wireless sensing. Therefore, the EM propagation characteristics are the primary concern in designing and evaluating the ISAC performance. 
    
    To satisfy the high data rate requirement in 6G, the high frequency and extremely large antenna array are two essential technologies, which bring a fundamental transition to the EM propagation, from far-field (FF) to near-field (NF) propagation. For example, the Rayleigh distance, which distinguishes the NF and FF region, is about ${d}_{Ray}=\frac{2{D}_{{\rm{BS}}}^{2}}{{\lambda}_{c}} \approx 88$m for the base station (BS) equipped with ${N}_{t}=128$ uniform linear array (ULA) at ${f}_{c}=28$ GHz. Different from the FF propagation, the EM propagation in the NF region evolved from the parallel wavefront to the spherical wavefront~\cite{NF_yw}. The mismatch between the conventional parallel wavefront model and NF propagation leads to a loss in the beamforming (BF) gain and the radio sensing resolution. 

    To precisely depict the NF propagation, the spherical wavefront model, where the phase shifts of the steering vectors are jointly determined by the angle and distance, provides promising advantages to both the communication and radio sensing functionalities. For the communication aspect, the induced distance information in spherical wavefront can be exploited to reduce the inter-user interference to improve the spectrum efficiency~\cite{ZidongWu_23JSAC,Haiyang_22TWC}. For the radio sensing aspect, the spherical wavefront extends one-dimensional (1D) estimation to two-dimensional (2D) or even three-dimensional (3D) estimation. The authors of~\cite{Alessio_22JSTSP} proposed a 2D estimation on azimuth and elevation angles to save the spectrum resources and elevate the radio sensing performance. The authors of~\cite{Haocheng_3DLocalization} investigated the 3D estimation on Cartesian coordinates via the separate Cram{\'e}r-Rao bound (CRB) metric. Furthermore, the benefits of NF in ISAC systems are still in the infancy stage, from ULA to uniform planar array (UPA). The authors of \cite{Zhaolin23_NF} investigated a downlink NF-ISAC system with ULA, which indicated that the beampattern can be focused on a region with the desired angle and desired distance without leakage in the undesired distance. The authors of \cite{Haochen_JSAC} investigated a joint uplink and downlink NF-ISAC framework with ULA, where the radio sensing performance can be elevated via the limited bandwidth. To exploit the degree of freedom (DoF) in UPA, the authors of~\cite{Boqun_JSTSP} conducted the performance analysis of the NF-ISAC on the accurate channel modelling to unveil the effect of effective antenna aperture and the polarization mismatch. The authors of \cite{Haochen_NF-ISAC} investigated the BF design of UPA in a downlink NF-ISAC system, which demonstrated that the optimal solutions had low-rank structures. 
    
    As the NF propagation relies on the equivalent aperture size, the antenna array design brings new challenges in NF-ISAC. The conventional antenna array design (ULA or UPA) cannot fully utilize the NF effect. On the one hand, the ULA or UPA cannot maintain the equivalent Rayleigh distance among different azimuth angles. When the sensing target (ST) is located at the edge of ULA, there is a great loss in the BF gain towards the ST. On the other hand, the UPA with $\frac{{\lambda}_{c}}{2} $ antenna interval requires massive antennas to support an extended Rayleigh distance because of the limited aperture size. The massive antennas bring a severe hardware burden at high frequencies. Fortunately, the uniform circular array (UCA), which has the inherent homogeneous characteristic, has drawn great attention to accommodate the NF propagation. Due to the rotational symmetry of the circular array, the UCA has the equivalent Rayleigh distance over different azimuth angles~\cite{Zidong_24TWC}. This characteristic has been exploited to improve ergodic capacity~\cite{24TCOM_26} and spectrum efficiency~\cite{Zidong_24TWC} as well as reduce the beam training overhead~\cite{Yuxin_23WCL}. The authors of~\cite{YJeon_21TWC} focused on the UCA-based line-of-sight (LoS) multiple-input-multiple-output (MIMO), where the channel matrix can be diagonalized via the discrete Fourier transform (DFT) precoder and the inverse DFT combiner. The isotropic characteristic of the induced beam pattern by UCA also enables the $360^{\circ}$ scan-angle range, where both the azimuth and elevation angles can be estimated~\cite{Mathews_94TSP}. The authors of \cite{Zhaolin23_NF} conducted the asymptotic analysis on the separate CRBs over angle and distance, which indicated that the CRBs with isotropic BF are irrelevant to the ST's direction. 
 
         \begin{table*}[!ht]\large
	\caption{Our contributions in contrast to the state-of-the-art.}  \label{Table:Comparition}
       \begin{tiny}
           \begin{center}
		\centering
		\resizebox{\textwidth}{!}{
			\begin{tabular}{ !{\vrule width1.2pt}l !{\vrule width1.2pt}c !{\vrule width1.2pt}c!{\vrule width1.2pt}c!{\vrule width1.2pt}c!{\vrule width1.2pt}c!{\vrule width1.2pt}c!{\vrule width1.2pt} }
				\Xhline{1.4pt}
				\centering
				& \cite{Kangjian_24JSTSP} & \cite{Ahmet_24JSTSP} & \cite{Xiangrong_24JSTSP} & \cite{Xinrui_24TWC} & \cite{Huizhi_arXiv} & \bf{Proposed} \\
				\Xhline{1.4pt}
				\centering
				Involved Functionality & Com-only & ISAC & ISAC &  Com-only  & Com-only  & ISAC \\
				\hline
				\centering
				2D Antenna Array  & $\times$ & $\times$ & $\surd$ & $\times$ & $\times$ & $\surd$ \\
				\hline
				\centering
			    Involved Coordinate  & $\times$ & 1D, DoA & 2D, Polar & $\times$ & $\times$ & 3D, Cylindrical\\
				\hline
				\centering 
				Impact of UCA on ISAC & $\times$ & $\times$ & $\times$ &$\times$& $\times$ & $\surd$ \\
				\hline
				\centering 
				Impact of BF design on joint radio sensing & $\times$ & $\surd$ & $\surd$ &$\times$& $\times$ & $\surd$ \\ 
				\Xhline{1.4pt}
		\end{tabular}  }
	\end{center}
       \end{tiny}
	\label{table:structure2}
        \end{table*}

\subsection{Motivation and Contributions} 
However, the aforementioned works mainly focus on the NF propagation in ISAC with conventional antenna array, how to exploit the antenna array to release the full potential of the NF propagation is still an open question to be unveiled. As the phase shift of spherical wavefront is simultaneously determined by the angle and distance, the antenna array structure plays a predominant role because of the aperture size and polarization mismatch~\cite{Boqun_JSTSP}. The effective array aperture size of conventional ULA or UPA dramatically reduced, which cannot maintain the NF propagation over different directions. Benefiting from the equivalent Rayleigh distance over different angles, UCA has shown its superiority in communication performance, which is a promising method to boost radio sensing performance. To the best of the authors' acknowledgment, there is a lack of work on the impact of UCA on the joint radio sensing performance in NF-ISAC systems. Meanwhile, the high computational complexity in NF-ISAC poses a significant constraint on its further development. The massive transmit antennas extremely increase the dimension of the BF design. The conventional algorithms usually optimize the BF design in a matrix form, which further increases the computation complexity. Thus, how to propose a low-complexity algorithm to unveil the optimal BF design structure is an important question for its practical application. 

Motivated by these points, we propose an UCA-based NF-ISAC framework, where the ST in coplanar or non-coplanar of the UCA are investigated to unveil the impact of antenna array geometry on the joint positioning performance. By exploiting the potential of spherical wavefront, we first derive the 3D joint squared position error bound (SPEB) approximation under the Cylindrical coordinate system. Inspired by the derived SPEB metric, the BF design is investigated to improve the joint positioning performance. The main contributions of this work are summarized below, which has been boldly and explicitly compared with the state-of-art in Table~\ref{Table:Comparition}.

\begin{itemize} 
     \item We propose a UCA-based NF-ISAC framework, where the ISAC BS is equipped with the UCA to simultaneously serve the ST and communication user (CU). To evaluate the joint 3D positioning performance, the SPEB approximation over the coplanar and non-coplanar STs are derived. A joint SPEB minimization problem is formulated to optimize the BF design while satisfying the communication and power consumption requirements.   

    \item For the coplanar case, where the ST is located in the co-plane of the UCA, a joint SPEB approximation is derived to evaluate the positioning performance. With the NF propagation, the Fisher information matrix (FIM) for the position vector degrades into a diagonal matrix. To evaluate the BF design on radio sensing performance, a joint SPEB minimization problem is formulated subject to the communication and power consumption requirements. A closed-form expression is first derived to optimize the seperate CRB of angle. A vector-based quadratic transformation (VQF) algorithm is proposed to optimize the BF design with the low-complexity, where the Rayleigh quotient is employed to decompose the partial derivative matrix. 

    \item For the non-coplanar case, where ST is located beyond the co-plane of the UCA, a joint 3D SPEB approximation is derived to evaluate the positioning performance. Unlike co-planar case, the FIM under non-coplanar case is a block-diagonal matrix due to the relevance between the perpendicular distance and the signed distance. In the same time, the radio sensing performance depends on the variation of ST. An extended VQF algorithm is proposed to improve the joint positioning performance constraints by the communication and power budget requirement.   

    \item Our numerical results validate that the derived SPEB approximation nearly reaches the identical performance with the numerical results. Meanwhile, the proposed VQF algorithms outperform the conventional semi-definite relaxation (SDR) algorithm in both performance precision and computation complexity. Furthermore, the UCA achieves better performance than UPA when ST is located in the coplanar of UCA and the perpendicular distance of ST is larger than its signed distance in non-coplanar case.   
   
\end{itemize}
	
\subsection{Organization and Notation}
The rest of this paper is organized as follows. In Section \RomanNumeralCaps{2}, the NF-ISAC systems with the UCA in both the coplanar case and non-coplanar case are presented. A joint SPEB minimization problem is formulated subject to the communication and power budget requirement. In Section \RomanNumeralCaps{3}, the joint SPEB performance under coplanar case is first analyzed, where the separate ${\rm{CRB}}$ approximation on azimuth angle and perpendicular distance are derived with the isotropic BF design. By exploiting the predominant component of ${\rm{CRB}}$, a VQF algorithm is proposed to optimize the BF design. In Section \RomanNumeralCaps{3}, the joint SPEB performance under the non-coplanar case is analyzed, where the joint SPEB can be approximated as a block diagonal FIM with negligible mismatch. An extended VQF algorithm is proposed to optimize the BF design. In Section \RomanNumeralCaps{5}, numerical results are provided to validate the derived radio sensing metric and evaluate the performance of our proposed algorithms. Finally, we conclude this paper in Section \RomanNumeralCaps{6}.  \\

\emph{Notations:} Scalars, vectors, and matrices are denoted by lower-case, bold-face lower-case, and bold-face upper-case letters, respectively; ${\mathbb{C}}^{N \times 1}$ denotes the set of $N \times 1$ complex-valued vectors; Symbol $\odot$ denotes the Hadamard (element-wise) product between two matrices (vectors);  $x^{*}$ denotes the complex conjugate of $x$; ${\mathbf{a}}^{\rm{H}}$, $\left\| {\mathbf{a}} \right\|_{l}$ represent the conjugate transpose and the $l, l \in \left\{1,2\right\}$ norm of vector $\mathbf{a}$; Rank(${\mathbf{A}}$) and Tr(${\mathbf{A}}$) denote the rank and trace of matrix ${\mathbf{A}}$, respectively. ${\mathbf{A}} \succeq 0$ means that ${\mathbf{A}}$ is a positive semidefinite matrix. 
	
	\begin{figure}[t!]
		\centering
		\subfigure[Coplanar Case.]{\label{wraparound}
			\includegraphics[width=1.5in]{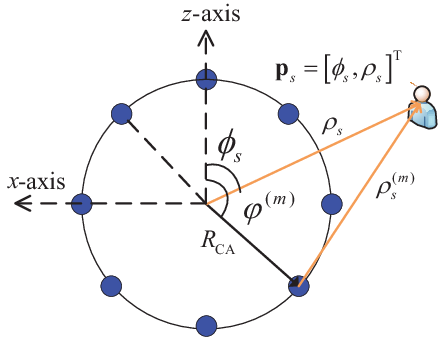}}
        \subfigure[Non Coplanar Case.]{\label{Fig:System_NF_ISAC_Antenna}
		\includegraphics[width=1.5in]{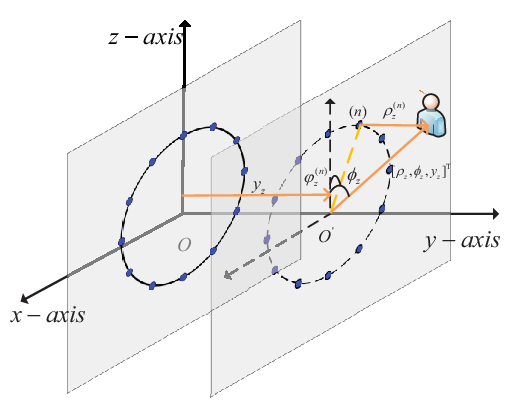} }
		\caption{The antenna array geometry pattern at the BS.}\label{Fig:SAA_pattern}
	\end{figure}

\section{System Model and Problem Formulation} 
   Considering that the phase of EM wavefront in NF simultaneously depends on the angle and distance, we first illustrate the system setting and NF propagation model following by the radio sensing and communication model. As illustrated in Fig.~\ref{Fig:SAA_pattern}, the antenna array at the BS is composed of $N$ UCA array, where both the ${N}_{t}$ transmit antennas and ${N}_{r}$ receive antennas are uniformly distributed along the circular with the radius ${R}_{{\rm{CA}}}$ respectively. Without loss of generosity, the central of the antenna array is chosen as the reference point with the coordination ${\mathbf{r}}_{0 }=\left[ 0,0,0\right]^{{\rm{T}}}$ and the clockwise is employed to model the antenna index. The corresponding coordinates of the $m$-th antennas are given by ${\mathbf{p}}_{{\rm{CA}}}^{\left(m\right)} =\left[ R_{{\rm{CA}}}, {\varphi}_{{\rm{CA}}}^{\left(m\right)},0 \right]^{{\rm{T}}}$, where $ {\varphi}_{{\rm{CA}}}^{\left(m\right)}=\frac{2m\pi}{ {N}_{l}}, l=r,t$. The corresponding Cylindrical coordinates of the ST (CU) is given by ${\mathbf{p}}_{i}=\left[ {\rho}_{i}, {\phi}_{i}, {y}_{i} \right]^{{\rm{T}}}, i \in \left\{ { s,c } \right\} $, where ${\rho}_{i}, {\phi}_{i}$ denote the perpendicular distance and the polar angle from the ST/CU to the original in the parallel $x-z$ plane while ${y}_{i} $ denotes the signed distance. To measure the impact of antenna geometry pattern on joint radio sensing performance, we compare the coplanar and non-coplanar ST as shown in Fig.~\ref{Fig:SAA_pattern}. 
   
\subsection{The NF model for UCA Geometry Pattern} 
Based on the different location of ST/CU, we take the ST as an example to investigate the coplanar case and non-coplanar case in the following.  
\subsubsection{Coplanar Case (${y}_{s}=0$) } 
For the coplanar case, the antenna elements are compactly distributed along the circular with the radius ${R}_{{\rm{CA}}}$ while the ST is located in the same plane (i.e., $y_{s}=0$) of the antenna array. For the non-coplanar case, the ST is not in the same plane (i.e., $y_{s}\neq 0$) of the antenna array. Considering that the phase shifts of spherical wave propagation simultaneously depends on angle and propagation distances between each antenna element and the ST, the corresponding steering vector ${\bm{\alpha}}\left( {\mathbf{p}}_{s} \right)$ between the $m$-th antenna and the ST is given by
    \begin{equation}\label{NearField_Steer}
       {\left[ {{\bm{\alpha}} \left( {\mathbf{p}}_{s} \right)} \right]_{\left( {m} \right)}} = {e^{ - j\frac{{2\pi }}{{{\lambda _c}}}{l^{\left( {m} \right)}_{l,i}} \left( {\mathbf{p}}_{s}  \right)}}, i=1,2, 
    \end{equation}where ${l}^{\left(m\right)}_{l,i} \left( {\mathbf{p}}_{s} \right)$ denotes the propagation distance between the $m$-th antenna and the ST. The distance between the ST and the $m$-th antenna at the UCA can be expressed as ${\rho}_{{l,1}}^{m}=\sqrt{ {R}_{{\rm{CA}}}^{2}+ {\rho}_{s}^{2}-2{\rho_{s}}{R}_{{\rm{CA}}} \cos\left({\theta}_{s} -{\varphi}_{{\rm{CA}}}^{\left(m\right)} \right) } $. By invoking the uniform spherical wavefront (USW) model, the corresponding propagation distance is given by ${l_{{l,1}}^{(m)}}( {\mathbf{p}}_{s} ) ={\rho}_{{l,1}}^{m}- {\rho_{s}}$. The corresponding steering vector can be expressed as $\left[{\bm{\alpha}}_{l,1}\right]_{m}={e^{ - j\frac{{2\pi }}{{{\lambda _c}}}{l_{l,1}^{\left( {m} \right)}} ({\mathbf{p}}_{s}) }}  $.
    
\subsubsection{Non-coplanar Case (${y}_{s} \neq 0$) }
Except for the special coplanar case, the non-coplanar case refers to a more general case where the ST is located beyond the co-plane of the UCA (${y}_{s} \neq 0$) as shown in Fig.~\ref{Fig:System_NF_ISAC_Antenna}\footnote{The special case where the ST is located along the normal direction of the antenna array is not considered in non-coplanar case. As the propagation environment along the normal direction of the antenna array degrades into the parallel wavefront due to symmetry characteristics and isotropic BF design, the NF propagation didn't hold for this case.}. The propagation distance is given by $l_{l,2}^{(m)}( {\mathbf{p}}_{s} ) = \sqrt {{{(\rho_{ {\rm{CA}} }^{(m)})^2} + y_s^2} } - \sqrt {\rho _s^2 + y_s^2} $, where $ {(\rho_{{\rm{CA}}}^{(m)})^2} = R_{{\rm{CA}}}^2 + {\rho}_s^2 - 2{R_{{\rm{CA}}}}{\rho _s}\cos ({\phi _s} - {\varphi}_s^{(m)}) $. The steering vector under non-coplanar case can be expressed as $ \left[{\bm{\alpha}}_{l,2}\right]_{m}={e^{ - j\frac{{2\pi }}{{{\lambda _c}}}{l_{l,2}^{\left( {m} \right)} \left( {\mathbf{p}}_{s} \right) } }} $. 

 \subsection{Radio Sensing Model}
 By denoting the joint communication and sensing (C\&S) signal as $c\left[n\right]$, the transmit signal from the BS is given by ${\mathbf{x}}\left[n\right] = {\mathbf{w}} {s}\left[n\right]$, where ${\mathbf{w}} \in {\mathbb{C}}^{{N}_{t} \times 1}$ denotes the BF design of joint C\&S signal. The received echo signal at the BS can be expressed as ${\mathbf{Y}}_{s}={\alpha}_{s} {\bm{\alpha}}_{r,i} \left( {\mathbf{p}}_{s} \right) {\bm{\alpha}}^{\rm{H}}_{t,i} \left( {\mathbf{p}}_{s} \right) {\bar{{\mathbf{X}}}}_{l}+ {\mathbf{N}}_{s} \in {\mathbb{C}}^{ {N}_{r} \times {N}_{t} } $, where ${\bar{{\mathbf{X}}}}_{l}= [ {\bar{\mathbf{x}}}[1], ..., {\bar{\mathbf{x}}}[L]]  \in {{\mathbb{}{C}}^{{N_t} \times L}}$ denotes the observed transmit signal over $L$ time slots and ${{\mathbf{N}}}_{s} \sim {\mathcal{CN}} \left(0, {\sigma}^{2} \right)$ is the additive white Gaussian noise (AWGN). The vectorized received echo signal is given by ${\mathbf{y}}_{s}={\rm{vec}}\left( {\mathbf{Y}}_{s} \right)= {\mathbf{m}}_{s}+{\mathbf{n}}_{s}$, where ${\mathbf{m}}_{s}={\rm{vec}} \left( {\alpha}_{s} {\bm{\alpha}}_{r} \left( {\rho}_{s},{\phi}_{s}, {y}_{s} \right) {\bm{\alpha}}^{\rm{H}}_{t} \left( {\rho}_{s},{\phi}_{s}, {y}_{s} \right) {\bar{{\mathbf{X}}}}_{l} \right) $ is the desired part in the received signal ${\mathbf{y}}_{s}$ at ST and ${\mathbf{n}}_{s}={\rm{vec}} \left( {\mathbf{N}}_{s}\right) $. ${\tilde{\bm{{\alpha}}}}_{s}=\left[ {\rm{Re}}\left({\alpha}_{s}\right), {\rm{Im}}\left({\alpha}_{s}\right) \right]$, the FIM ${\mathbf{J}}_{\bm{\gamma}} \in {\mathbb{R}}^{4 \times 4} $ of the estimated parameter vectors ${\bm{\gamma}}=\left[ {\mathbf{p}}_{s}  \; {\tilde{\bm{{\alpha} }}}_{s} \right]^{{\rm{T}}}$ is given by ${\mathbb{E}} \left\{ ( {\hat{{\bm{\gamma}}}} - {\bm{{\gamma}}}) ( {\hat{{\bm{\gamma}}}} - {{\bm{\gamma}}})^{{\rm{T}}} \right\} -  {\mathbf{J}}_{\bm{\gamma}}^{-1} \succeq 0 $, where ${\hat{{\bm{\gamma}}}}$ denotes the unbiased estimator of ${\bm{{\gamma}}}$. As the observations ${\mathbf{y}}_{s}$ follows the complex Gaussian distribution, the $\left(i,j\right)$-th FIM elements $\left[ {\mathbf{J}}_{\bm{\gamma}} \right]_{i,j} \buildrel \Delta \over = \Lambda(\gamma_i,\gamma_j) $ can be expressed using the Slepian-Bangs formula as~\cite{SPEB_0}
   \begin{equation}\label{FIM_element0}
   	\Lambda({\gamma}_{i}, {\gamma}_{j}) = \frac{2}{\sigma^2} \sum_{n=0}^{L-1} \Re\left\{\left( \frac{\partial m_{s}[n]} {\partial \gamma_i} \right)^{*} \frac{\partial m_{s}[n]} {\partial \gamma_j}\right\}.  
   \end{equation}

\subsection{Communication Model} 
For communication functionality, the joint C\&S beam is exploited to convey the modulated information signal towards the CU. Let $s[n] $ denotes the joint C\&S signal, the transmitted signal from the ISAC BS is given by ${\mathbf{x}}[n]={\mathbf{w}}s[n]$, where ${\mathbf{w}} \in {\mathbb{C}}^{{N}_{t} \times 1}$ denotes the BF design for the joint C\&S signal. The covariance matrix of the transmit signal is given by ${\mathbf{R}}_{x}= \frac{{\bar{\mathbf{X}}} {{\bar{\mathbf{X}}}^{\rm{H}}}}{L}$. The power budget of the transmitted signal can be expressed as ${P}={\mathbb{E}} ( {\mathbf{x}}^{{\rm{H}}}[n] {\mathbf{x}}[n])={\mathbf{w}}^{{\rm{H}}} {\mathbf{w}}$. The achievable date rate can be expressed as $ {R}_{c}=\log_{2} (1+\frac{ {\left| {{\mathbf{h}}_c^{{\rm{H}}} {\mathbf{w}}}  \right|}  }{{\sigma}^{2}})$. ${\mathbf{h}}_{c} \in {\mathbb{C}}^{{N}_{t} \times 1}$ denotes the NF communication channel between the ISAC BS to CU with the expression ${\mathbf{h}}_{c}= {\bm{\beta}} ({r}_{c}) \odot {\bar{{\bm{\alpha}}}}_{c,l} $, where ${\bm{\beta}} ({r}_{c}) $ denotes the amplitude gain with the expression $[{\bm{\beta}} ({r}_{c})]_{m}=\frac{{\lambda}_{c}}{ 4\pi {r}_{c}^{m}({\mathbf{p}}_{c}) } $ and ${\bar{{\bm{\alpha}}}}_{c}$ denotes the steering vector between the ISAC BS and the CU with the expression $[ {\bar{{\bm{\alpha}}}}_{c} ]_{m}=e^{-j\frac{2\pi l_c^{(m)}\left( {\mathbf{p}}_{c} \right)  }{{\lambda}_{c}}} $. Similar to the steering vector of ST, the propagation distance $l_{c}^{(m)} \left( {\mathbf{p}}_{c} \right) $ between the $m$th antennas and CU is given by
\begin{equation}
    l_c^{(m)}\left( {\mathbf{p}}_{c} \right) = \left\{ \begin{array}{l}
          {\rho}_{{\rm{CA}},c}^{(m)} - {\rho _c},  \quad {y_c} = 0\\
\sqrt {{{\left( { {\rho}_{{\rm{CA}},c}^{(m)}} \right)}^2} + y_c^2}  - \sqrt {\rho _c^2 + y_c^2}, \quad  {y_c} \neq 0,
\end{array} \right.
\end{equation}where ${\rho}_{{\rm{CA}},c}^{(m)}=\sqrt{ R_{ {\rm{CA}} }^{2}+{\rho}_{c}^{2}-2{\rho}_{c} {R}_{ {\rm{CA}} } \cos\left( {{\phi _c} - {\varphi}_{{\rm{CA}}}^{(m)}} \right) }.$ 

\subsection{Problem Formulation}
To evaluate the position performance over multiple unknown parameters, the joint SPEB is given by~\cite{SPEB_0}
\begin{equation}\label{SPEB_def}
   {\rm{SPEB}} \left( {\mathbf{R}}_{x}; {\bm{\eta}} \right) \buildrel \Delta \over = {\rm{tr}} \left( {\mathbf{J}}_{\bm{\eta}}^{-1} \right)\leq  {\mathbb{E}}\big\{ \left\|{ {\hat{ {\bm{\eta}} }} - {\bm{\eta}} }\right\|^2  \big\},
\end{equation}where ${\bm{J}}_{ {\bm{\eta}}}$ denotes the FIM of the Cartesian coordinate ${\bm{\eta}}=\left[ {x}_{s},{z}_{s}, {y}_{s} \right]^{{\rm{T}}}$ and ${{\hat{ {\bm{\eta}} }}} $ denotes the unbiased estimator of ${\bm{\eta}}$. To improve the joint position performance of ST in the NF-ISAC system, the BF design problem is formulated as 
  \begin{subequations}\label{Eq:PF_ori}
  	\begin{align}
  		&\label{min PEB} \underset{ {{\mathbf{w}}} } {\mathop{\min }}  \quad {\rm{SPEB}}_{i}, i=1,2   \\    
  		\label{Communication QoS_NOMA_CU}& {\rm{s.t.}}\quad  \; {R}_{c} \ge {\bar{R}}_{\min}, \\
  		\label{Power max}& \quad \quad \quad   {\mathbf{w}}^{\rm{H}} {\mathbf{w}}  \le  P_{max}, 
  	\end{align}
  \end{subequations}where ${\rm{SPEB}}_{i}, i=1,2 $ denote the joint SPEB in coplanar ${\rm{SPEB}}_{1}$ and non-coplanar case ${\rm{SPEB}}_{2}$. ${\bar{R}}_{\min} $ denotes the minimal communication requirement of CU and \eqref{Power max} denotes the transmission power budget at the BS. Due to the weighted fractional form in \eqref{min PEB}, problem~\eqref{Eq:PF_ori} is a non-convex problem, which makes it hard to find the optimal solution. 

\section{Radio Sensing Performance Analysis and BF Design for Coplanar Case}
To unveil the impact of UCA on radio sensing performance in the NF-ISAC system, we first conduct the performance analysis on separate CRB and joint SPEB, where the closed-form approximation of CRB is derived with the isotropic BF design. Based on the derived seperate CRB, the structure of optimal solution is exploited, which inspire the Rayleigh quotient-based BF design to optimize the radio sensing performance.

\subsection{Radio Sensing Performance Analysis with Coplanar Case} 
To unveil the effect of the UCA on CRB performance, we first provide the radio sensing performance analysis for coplanar case, where the ST is located in the coplanar of the UCA. The CRB performance under the USW channel model is derived as shown in the following. 

\subsubsection{The CRB Approximation with USW Channel Model} 
By denoting ${\left[ {{\mathbf{v}}_{t, 1 }} \right]_m} = \frac{{{R_{{\rm{CA}}}}{\rho_s}\sin (\frac{{2\pi m}}{N_{t}} - {\phi _s})}} { {{  {(R_{ {\rm{CA}} }^2 + \rho _s^2 - 2{R_{{\rm{CA}}}}{\rho_s}\cos (\frac{{2\pi m }}{N_{t}} - {\phi _s}))}^{\frac{1}{2}} }} } $ and $\left[ {{\mathbf{v}}_{t,2}} \right]_{m}  = \frac{{({\rho_s} - {R_{{\rm{CA}}}}\cos (\frac{{2\pi m}}{N_{t}} - {\phi_s}))}}{{{{(R_{{\rm{CA}}}^2 + \rho_s^2 - 2{R_{{\rm{CA}}}}{\rho _s}\cos (\frac{{2\pi m}}{N_{t}} - {\phi _s}))}^{\frac{1}{2}}}}}{\rm{ - 1}}$, the derivatives ${\dot{{\bm{\alpha}}} }_{l,1,i}, l \in \left\{r,t\right\}, i \in \left\{1,2\right\}$ of the steering vector ${\bm{\alpha}}_{l,1}$ can be expressed as
    \begin{equation}\label{UCA_steer}
    	{\dot{{\bm{\alpha}}} }_{l,1,i} = \frac{{\partial {\bm{\alpha}} _{ {\rm{CA}},l} \left( {{\rho_s},{\phi_s}} \right)}} {{\partial {u_i}}} = \frac{{j2\pi }}{{{\lambda_c}}} {{\mathbf{v}}_{l, i}} \odot {\bm{\alpha}}_{l,1} \left( {\mathbf{p}}_{s} \right),
    \end{equation}where ${u}_{1}={\rho}_{s}, {u}_{2}={\phi}_{s}$. 
   
   \begin{proposition}\label{Pro:CRB_USW_UCA}
     Based on the USW model of the steering vector, the FIM ${\mathbf{J}}_{{\bm{\gamma}} } $ in coplanar case can be approximated as a diagonal matrix $ \left( { {\mathbf{J}}_{{\bm{\gamma}} } \approx {\mathbf{J}}_{{\bm{\gamma}} }^{{\rm{d}}} } \right) $ with the slight difference. By denoting ${\bm{\Upsilon}}(\alpha ) = \int\limits_{0}^{2\pi } {\frac{{\alpha  - \cos x}}{{2\pi \sqrt {1 - 2\alpha \cos x + {\alpha ^2}} }}} dx$, the corresponding separate CRB approximations can be approximated by 
   	\begin{subequations}\label{CRB_USW_C1}
   		\begin{equation}\label{CRB_USW_phi}
   			{\rm{CRB}}_{{\phi_s}} = \frac{{{\sigma ^2}}} {{2{{\left| {{\alpha _s}} \right|}^2}L\Re [ {\rm{Tr}} ( {{\mathbf{R}}_{x}} {\dot{{\mathbf{A}}}}_{1,2}^{\rm{H}} { {\dot{\mathbf{A}}}_{1,2} })]}},
   		\end{equation}
   	    \begin{equation}\label{CRB_USW_phi}
   	    	{\rm{CRB}}_{{\rho_s}} = \frac{{{\sigma ^2}}} {{2{{\left| {{\alpha _s}} \right|}^2}L\Re [ {\rm{Tr}} ( { {\mathbf{R}}_{x} } {\dot{{\mathbf{A}}}}_{1,1}^{\rm{H}} { {\dot{\mathbf{A}}}_{1,1} })]}},
   	    \end{equation}
   	\end{subequations}
   where 
   \begin{equation}\label{AA_phi_final}
   	    {\dot{{\mathbf{A}}}}_{1,2}^{\rm{H}} {\dot{{\mathbf{A}} } }_{1,2} = \frac{{4{\pi^2}} {N_r}}{{\lambda _c^2}} \left[ \left\| {{{\mathbf{v}}_{r,1}}} \right\|_2^2 {{\bar{{\mathbf{A}}}}_t}   + {\rm{diag}}({{\mathbf{v}}_{t,1}}){{\bar{{\mathbf{A}}}}_t} {\rm{diag}}({{\mathbf{v}}_{t,1}}) \right], 
   \end{equation}
   \begin{equation}\label{AA_rho_final_0}
   	   \begin{split}
   	   	   & {\dot{{\mathbf{A}}}}_{1,1}^{\rm{H}} {\dot{{\mathbf{A}} } }_{1,1} = \frac{{4{\pi^2}{N_r}}}{{\lambda _c^2}} \left[ {\rm{diag}}({{\mathbf{v}}_{t,2}}) {{\bar{{\mathbf{A}}}}_t} {\rm{diag}}({{\mathbf{v}}_{t,2}}) \right]  \\
   	   	   & -\frac{ {4{\pi^2}{N_r}} } {\lambda_c^2 } \left[ {\bm{\Upsilon}}(\frac{{{\rho _s}}}{{{R_{{\rm{CA}}}}}}) - 1 \right] [{\rm{diag}} ({{\mathbf{v}}_{t,2}}) {{\bar{{\mathbf{A}}}}_t} + {{\bar{{\mathbf{A}}}}_t} {\rm{diag}}({{\mathbf{v}}_{t,2}})] \\
   	   	   & + \frac{{4{\pi^2}{N_r}}}{{\lambda _c^2}} \left\| {{{\mathbf{v}}_{r,2}}  } \right\|_2^2 {{\bar{{\mathbf{A}}}}_t}.  
   	   \end{split}
   \end{equation}$\left\| {{{\mathbf{v}}_{l,2}}  } \right\|_2^2= {N}_{l} \left[ \frac{3}{2} - 2 {\bm{\Upsilon}} (\frac{{{\rho _s}}}{{{R_{{\rm{CA}}}}}}) \right] $ if ${\rho}_{s} <{R}_{{\rm{CA}}} $ while $\left\| { {{\mathbf{v}}_{l,2}}  } \right\|_2^2= {N}_{l} \left[  2 - \frac{{R_{{\rm{CA}}}^2 {N}_{l} }}{{2\rho _s^2 }} - 2 {\bm{\Upsilon}} (\frac{{{\rho_s}}}{{{R_{{\rm{CA}}}}}}) \right] $ if ${\rho}_{s} > {R}_{{\rm{CA}}} $ .
   With the isotropic BF design at the BS, the CRB of ${\phi}_{s}$ can be further simplified as:
   \begin{subequations}
   	  \begin{equation}\label{CRB_USW_iso}
   	  	{\rm{CRB}}_{{\phi_s}} = \left\{ \begin{array}{l}
   	  		\frac{{ {\lambda}_c^2 {\sigma^2} }}{{8{{\left| {{\alpha _s}} \right|}^2}{\pi ^2}\rho_s^2{N_r}{P_{\max}}L}}, \quad {\rho _s} \le {R_{{\rm{CA}}}}, \\
   	  		\frac{{ {\lambda}_c^2 {\sigma ^2} }}{{8{{\left| {{\alpha _s}} \right|}^2}{\pi ^2}R_{{\rm{CA}}}^2{N_r}{P_{\max }}L}},{\rm{  }}{\rho _s} > {R_{{\rm{CA}}}},
   	  	\end{array} \right. 
   	  \end{equation}
     \begin{equation}\label{CRB_rho_USW_iso}
     	{\rm{CRB}}_{{\rho_s}}= \left\{ \begin{array}{l}
     		\frac{ {\lambda _c^2{\sigma ^2}} } { {8{{\left| {{\alpha _s}} \right|}^2}{\pi ^2} {N_r}{P_{\max }}L} \left[ 1-2\bm{\Upsilon}^{2} \left( \frac{ {\rho}_{s} }{ {R}_{{\rm{CA}}} } \right)  \right] },  {\rho _s} \leq {{{R}}_{{\rm{CA}}}}, \\
     		\frac{{ {\lambda}_c^2 {\sigma ^2} }} {{8{{\left| {{\alpha _s}} \right|}^2}{\pi ^2}{N_r}{P_{\max }}L \left[ 2- \frac{ {R}_{{\rm{CA}}}^{2} }{ {\rho}_{s}^{2} }-2 {\bm{\Upsilon}}^{2} \left( \frac{ {\rho}_{s} }{ {R}_{{\rm{CA}}} } \right)  \right]  }} \; {\rho _s} > {{{R}}_{{\rm{CA}}}}. 
     	\end{array} \right.
     \end{equation}
   \end{subequations}

   \end{proposition}
   \begin{proof}
   	Let ${\mathbf{A}}_{1} ={\bm{\alpha}}_{r,1} {\bm{\alpha}}_{t,1}^{\rm{H}}$, the corresponding derivatives of $\left\{{\phi}_{s}, {\rho}_{s}\right\} $ on ${\mathbf{A}}_{1}$ can be expressed as
   	  \begin{subequations}\label{A_derivative}
              \begin{equation}
   			{{\dot{{\mathbf{A}}}}_{1,{\rho_{s}}}}= \frac{{\partial{{\mathbf{A}_1}}}} {{\partial {\rho_{s}}}}  =\frac{{j2\pi }}{{{\lambda _c}}}\left[ 
                 {{{\mathbf{A}}_{{\rm{1}}}} {\rm{diag}}({{\mathbf{v}}_{t,2}}) - {\rm{diag}}({{\mathbf{v}}_{r,2}}) {{\mathbf{A}}_{{\rm{1}}}}} \right].
   		 \end{equation}
   		\begin{equation}
   			{{\dot{{\mathbf{A}}}}_{1,{\phi_{s}}}} = \frac{{\partial{{\mathbf{A}_1}}}} {{\partial {\phi_{s}}}} = \frac{{j2\pi }}{{{\lambda_c}}}\left[ {{{\mathbf{A}}_{{\rm{1}}}} {\rm{diag}}({ {\mathbf{v}}_{t,1}}) - {\rm{diag}}({{\mathbf{v}}_{r,1}}){{\mathbf{A}}_{{\rm{1}}}}} \right],
   		\end{equation}
   	\end{subequations}
    
    \begin{figure*}
    \begin{small}
        \begin{equation}\label{AA_deri_ori}
   		\begin{split}
   			{\dot{{\mathbf{A}}}}_{1,{\phi_s} }^{{\rm{H}}} {\dot{{\mathbf{A}}}}_{1,{\phi_s}} & =  \frac{{4{\pi ^2}}} {{\lambda_c^2}} \left[ {\mathbf{A}}_{1}^{{\rm{H}}} {\rm{diag}}({{\mathbf{v}}_{r,1}} \odot {{\mathbf{v}}_{r,1}}) {{\mathbf{A}}_{1}} +
   			{\rm{diag}}({{\mathbf{v}}_{t,1}}) {\mathbf{A}}_{1}^{{\rm{H}}} {{\mathbf{A}}_{1}} {\rm{diag}}({{\mathbf{v}}_{t,1}}) \right]  \\
   			& - \frac{{4{\pi ^2}}} {{\lambda_c^2}} \left[{\rm{diag}}({{\mathbf{v}}_{t,1}}) {\mathbf{A}}_{1}^{{\rm{H}}} {\rm{diag}}({{\mathbf{v}}_{r,1}}) {{\mathbf{A}}_{1}}  + {\mathbf{A}}_{1}^{{\rm{H}}} {\rm{diag}}({{\mathbf{v}}_{r,1}}){{\mathbf{A}}_{1}} {\rm{diag}}({{\mathbf{v}}_{t,1}})  \right]. 
   		\end{split}
   	\end{equation}
   	\hrulefill
    \end{small}
   \end{figure*}
   
   \noindent We take $\left[{\mathbf{J}}_{{\bm{\gamma}}} \right]_{(1,2)} $ for example to illustrate the approximation of the cross term in FIM. $\left[{\mathbf{J}}_{{\bm{\gamma}}} \right]_{(1,2)} $ can be expressed as $\left[{\mathbf{J}}_{{\bm{\gamma}}} \right]_{(1,2)} =\frac{2 {\left| {{\alpha_s}} \right|^2}L  }{{\sigma}^{2}  } {\rm{Tr}} \left( {\mathbf{R}}_{x} {\dot{{\mathbf{A}}}_{1,{\phi}_{s}}^{\rm{H}} {\dot{{\mathbf{A}}}}_{1,{\rho}_{s}} } \right) $, where each terms in ${\dot{{\mathbf{A}}}_{1,{\phi}_{s}}^{\rm{H}} {\dot{{\mathbf{A}}}}_{1,{\rho}_{s}} } $ contains ${\mathbf{v}}_{l,1} \odot {\mathbf{v}}_{l,2}$. It is noted that $ \left[ {{\mathbf{v}}_{t,1}}  \odot {{\mathbf{v}}_{t,2}} \right]_{m}=\frac{ {[{R_{{\rm{CA}}}} \rho _s\sin (\frac{{2\pi m}}{{{N_t}}} - {\phi _s})]} \left[ {\rho}_{s}+ {R}_{{\rm{CA}}} \cos(\frac{{2\pi m}}{{{N_t}}} - {\phi _s} ) \right] } { {(R_{{\rm{CA}}}^2 + {\rho}_s^2 - 2{R_{{\rm{CA}}}} {\rho_s}\cos (\frac{{2\pi m}}{{{N_t}}} - {\phi_s}))} }  - \frac{{[{R_{{\rm{CA}}}}{\rho _s}\sin (\frac{{2\pi m}}{{{N_t}}} - {\phi _s})]}} {R_{{\rm{CA}}}^2 + \rho_s^2 - 2{R_{{\rm{CA}}}}{\rho_s}\cos ( \frac{{2\pi m}}{{{N_t}}}-{\phi}_{s} ) }.$ Considering that $ \sum\limits_{m =  - \frac{{N_{t} - 1}}{2}}^{m = \frac{{N_{t} - 1}}{2}} {\frac{{[{R_{{\rm{CA}}}} \rho_{s} \sin (\frac{{2\pi m}}{{{N_t}}} - {\phi _s})]}}{{(R_{{\rm{CA}}}^2 + \rho _s^2 - 2{R_{{\rm{CA}}}}{\rho _s}\cos (\frac{{2\pi m}}{{{N_t}}} - {\phi_s}))}}} \approx 0$, we have  $\sum\limits_{m =  - \frac{{{N_t} - 1}}{2}}^{m = \frac{{{N_t} - 1}}{2}} \left[ {{\mathbf{v}}_{t,1}}  \odot {{\mathbf{v}}_{t,2}} \right]_{m} \approx 0$, which indicates that ${\mathbf{J}}_{{\bm{\gamma}}} \approx {\mathbf{J}}_{{\bm{\gamma}}}^{{\rm{d}}} $, ${\rm{CRB}}_{{\phi}_{s}} \approx \frac{1}{{{{\left[ {\mathbf{J}}_{{{\bm{\gamma}}}} \right]}_{\left(1,1\right) }}}} $ and ${\rm{CRB}}_{{\rho}_{s}} \approx \frac{1}{{{{\left[ {\mathbf{J}}_{{{\bm{\gamma}}}} \right]}_{\left(2,2\right) }}}}  $. While the $ {\dot{{\mathbf{A}}}}_{1,{\phi_s}}^{{\rm{H}}} {\dot{{\mathbf{A}}}}_{1,{\phi_s}} $ can be expressed as \eqref{AA_deri_ori}. By exploiting the periodical property of ${\mathbf{v}}_{r,1}$, we have ${\mathbf{A}}_{1}^{{\rm{H}}} {\rm{diag}}({{\mathbf{v}}_{r,1}}) {\mathbf{A}}_{1}=\sum\limits_{m=- \frac{{N}_{r}-1}{2} }^{ m= \frac{{N}_{r}-1}{2} } {{[{{\mathbf{v}}_{r,1}}]}_m} {{\bar{{\mathbf{A}}}}_{1,t}} \approx 0 $, where ${{\bar{{\mathbf{A}}}}_{t}}={\bm{\alpha}}_{t,1} {\bm{\alpha}}_{t,1}^{{\rm{H}}} $. By introducing ${\delta}=\frac{ 2 \pi}{N_{r}}$, the first term in \eqref{AA_deri_ori} can be derived as ${\mathbf{A}}_{1}^{{\rm{H}}} {\rm{diag}}( {{\mathbf{v}}_{r,1}} \odot {{\mathbf{v}}_{r,1}} ) {{\mathbf{A}}}_{1}= \left\| {{{\mathbf{v}}_{r,1}}} \right\|_2^2 = \sum\limits_{m = - \frac{{{N_r} - 1}}{2}}^{\frac{{{N_r} - 1}}{2}} \left[ {{\mathbf{v}}_{r,1}} \odot {{\mathbf{v}}_{r,1}}  \right]_{m} {\bar{{\mathbf{A}}}}_{t} $.
   \begin{equation}
   	  \begin{split}
   	  	 & \quad \sum\limits_{m = - \frac{{{N_r} - 1}}{2}}^{\frac{{{N_r} - 1}}{2}} \left[ {{\mathbf{v}}_{r,1}} \odot {{\mathbf{v}}_{r,1}}  \right]_{m}  \\
   	  	 & = \frac{ {R}_{{\rm{CA}}}^{2} {\rho}_{s}^{2} } { {\delta} } \sum\limits_{m = - \frac{{{N_r} - 1}}{2}}^{\frac{{{N_r} - 1}}{2}} { \frac{ {\delta} {{\sin}^2}(\frac{{2\pi m}}{N_{r}} - {\phi_s})  } { { { (R_{{\rm{CA}}}^2 + {\rho}_s^2 - 2{R_{{\rm{CA}}}}{\rho_s} \cos ( m {\delta} - {\phi_s})) } }} }  \\
   	  	 & \mathop  \approx \limits^{(a)} \frac{ {R}_{{\rm{CA}}}^{2} {\rho}_{s}^{2} {N}_{r} }{2 {\pi} } \int_{-\pi}^{\pi} \frac{ {{\sin}^2}x }{ R_{{\rm{CA}}}^2 + {\rho}_s^2 - 2{R_{{\rm{CA}}}}{\rho_s} \cos x } d x  \\
   	  	 & \mathop  \approx \limits^{(b)}  \left\{ \begin{array}{l}
   	  	 	\frac{{\rho_s^2{N_r}}}{2}, {R_{{\rm{CA}}}} \ge {\rho_s}\\
   	  	 	\frac{{R_{{\rm{CA}}}^2{N_r}}}{2}, {R_{{\rm{CA}}}} < {\rho _s},
   	  	 \end{array} \right.,
  	  \end{split}
   \end{equation}where $\mathop \approx \limits^{(a)} $ comes from ${\delta} \gg 1, {N}_{r} \gg 2\pi $ and $x=\frac{{2\pi m}}{N_{r}} - {\phi_s}$. $\mathop \approx \limits^{(b)} $ comes from the integral formula 3.613.3 in~\cite{Integral_reference}. Substituting the following equations into \eqref{AA_deri_ori}, \eqref{AA_phi_final} can be derived. Considering that ${\rm{Tr}}\left( {\bar{{\mathbf{A}}}}_{l} \right)={N}_{l} $ and ${\rm{Tr}} \left[ {\rm{diag}}\left( {\mathbf{v}}_{l,1} \right) {\bar{{\mathbf{A}}}}_{t} {\rm{diag}} \left( {\mathbf{v}}_{l,1} \right) \right]=\sum\limits_{m = - \frac{{{N_l} - 1}}{2}}^{\frac{{{N_l} - 1}}{2}} \left[ {{\mathbf{v}}_{l,1}} \odot {{\mathbf{v}}_{l,1}}  \right]_{m} $, \eqref{CRB_USW_iso} can be derived by substituting $ {\mathbf{R}}_{x}= \frac{P_{\max} }{{N}_{t}} {\mathbf{I}}_{{N}_{t} \times {N}_{t}}$ into \eqref{CRB_USW_phi}.
   \begin{lemma}\label{Lemma_Case1_norm}
     By exploiting the summation property of ${\mathbf{v}}_{l,2}$, we have 
   	  \begin{subequations}\label{v_2_int}
   	  	\begin{equation}
   	  		{\left\| { { {\mathbf{v}}_{l,2}} } \right\|_1} = {N_l} {\bm{\Upsilon}} (\frac{{{\rho_s}}}{{{R_{{\rm{CA}}}}}}) - {N_l}, 
   	  	\end{equation}
     	\begin{equation}\label{v_2_int_2}
     		 {\left\| { { {\mathbf{v}}_{l,2}} } \right\|_{2}^{2}}
     		= \left\{ \begin{array}{l}
     			\frac{{2\rho _s^2{N_r}}}{{(R_{{\rm{CA}}}^2 - \rho _s^2)}},  \; {\rho _s} < {R_{{\rm{CA}}}}\\
     			\frac{{2R_{{\rm{CA}}}^2{N_r}}}{{(\rho _s^2 - R_{{\rm{CA}}}^2)}}, \; {\rho _s} > {R_{{\rm{CA}}}}
     		\end{array} \right. .
     	\end{equation}
   	  \end{subequations}
   \end{lemma}
   \begin{proof}
   	 Please see the proof in {\textbf{Appendix \ref{App:Lemma_Case1_norm}}}. 
   \end{proof}
   \noindent Based on \eqref{v_2_int} and following the similar proof of ${\rm{CRB}}_{{\phi}_{s} } $, \eqref{CRB_rho_USW_iso} can be derived.
   \end{proof}

\subsection{Problem Formulation for Coplanar Case} 
According to {\textbf{Proposition~\ref{Pro:CRB_USW_UCA}}}, the FIM can be approximated as a diagonal matrix when the ST is located in the coplanar of the antenna array at the BS. Considering the radio sensing functionality is conducted to estimate the Cylindrical coordinate ${\mathbf{p}}_{s}$ of the ST, the projection matrix ${\mathbf{T}} \buildrel \Delta \over = \frac{\partial {\mathbf{p}}_{s}^{\mathsf{T}}} {\partial \bm{\eta}} \in {\mathbb{C}}^{3 \times 3}$ from the Cylindrical coordinate to the Cartesian coordinate is given by 

\begin{small}
    \begin{align} 
	{\mathbf{T}}  =\begin{bmatrix}
		\frac{{\partial {\rho_{s}}}} {{\partial {x_s}}}  & \frac{{\partial {\rho_{s }}}}{{\partial {y_s}}}  & \frac{{\partial {\rho_{s}}}}{{\partial {y_s}}} \\
		\frac{{\partial {\phi_{s}}}} {{\partial {x_s}}}   & \frac{{\partial {\phi_{s}}}} {{\partial {y_s}}} & \frac{{\partial {\phi_{s}}}}{{\partial {y_s}}} \\
        \frac{{\partial {y_{s}}}} {{\partial {x_s}}}   & \frac{{\partial {y_{s}}}} {{\partial {y_s}}} & \frac{{\partial {y_{s}}}}{{\partial {y_s}}} \\
	\end{bmatrix} 
    = \begin{bmatrix}
    \frac{ {x}_{s} } { \sqrt{{x}_{s}^{2}+{z}_{s}^{2} } },   & 	\frac{ {z}_{s} } { \sqrt{{x}_{s}^{2}+{z}_{s}^{2} } }, & 0 \\
		\frac{ -{z}_{s} } { {x}_{s}^{2}+{z}_{s}^{2} },  &  \frac{ {x}_{s} }{ {x}_{s}^{2}+{z}_{s}^{2} }, & 0 \\
		0,   & 	0, & 1 \\
	\end{bmatrix}.
\end{align}
\end{small}

\noindent Thus, the corresponding SPEB in coplanar case can be simplified by 
\begin{equation}
    {\rm{SPEB}} \left( {\mathbf{R}}_{x} ; {\bm{\eta}} \right)  =\frac{ {p}_{x}^{2} {p}_{y}^{2} }{ \left({p}_{x}+{p}_{y}\right)^{2} } \sum\limits_{i = 1}^2 \left\{  { [ {\mathbf{J}}_{ {{{\mathbf{p}}}_s} } ]}_{(i,i)}^{-1} { {[ {\tilde{{\mathbf{T}}}} ]}_{(i,i)} } \right\},
\end{equation}where $ {\tilde{{\mathbf{T}}}}=\frac{ {{{({p_x} + {p_y})}^2}} } { {p_x^2p_y^2} } {{{\mathbf{T}}^{-1}}{{({{\mathbf{T}}^{{\rm{H}}}})}^{ - 1}}}$. As it has proved that the FIM ${\mathbf{J}}_{{\bm{\gamma}}}$ can be approximated via the diagonal matrix ${\mathbf{J}}_{{\bm{\gamma}}}^{{\rm{d}}} $ with the trivial loss, the problem \eqref{Eq:PF_ori} can be further recast as
\vspace{-0.1cm}
\begin{small}
      \begin{subequations}\label{Eq:PF_ori_Case1}
  	\begin{align}
  		&\label{min PEB} \underset{ {\mathbf{w}}  } {\mathop{\min }}  \quad  \frac{ {{ { [ {\tilde{{\mathbf{T}}}} ]}_{(1,1)} }} } {\Re [{\rm{Tr}}( {\mathbf{w}}  {\dot{{\mathbf{A}}}_{1,{\phi}_s }^{{\rm{H}}} {{ {\dot{{\mathbf{A}}}}}_{ 1,{\phi}_s}{\mathbf{w}}^{{\rm{H}}} ) } }   ] } + \frac{{{{[ {\tilde{{\mathbf{T}}}} ]}_{(2,2)}}}} {{\Re [ {\rm{Tr}}( {\mathbf{w}}  {\dot{ {\mathbf{A}} }}_{1,{{\bm{\rho}}_s}}^{{\rm{H}}} {{{\dot{{\mathbf{A}} } }_{1,{{\bm{\rho}}_s} }} {\mathbf{w}}^{{\rm{H}}} ) } ] }} \\  
  		\label{Communication QoS_NOMA_CU} & {\rm{s.t.}} \quad  \; {R}_{c} \ge {\bar{R}}_{\min}, \\
  		\label{Power max}& \quad \quad \quad  {\mathbf{w}}^{\rm{H}} {\mathbf{w}}\le  P_{\max}.
  	\end{align}
    \vspace{-0.3cm}
  \end{subequations}
\end{small}

\noindent The problem~\eqref{Eq:PF_ori_Case1} is non-convex fractional programming because of the quadratic form, where the conventional Dinkelbach method~\cite{Dinkelbach} cannot be invoked due to the multi-ratio objective function (OF) in \eqref{Eq:PF_ori_Case1}. To solve this problem, we first focus on the first term to investigate the structure of the optimal solution and then employ the quadratic transformation~\cite{QuadraticTransformation} method in the following.
\vspace{-0.5cm}
\subsection{BF Design for Coplanar Case}

\subsubsection{BF Design to minimize ${\rm{CRB}}_{{\phi}_{s}}$ } 
To unveil the structure of BF design, we first investigate the BF design to minimize the first term in \eqref{Eq:PF_ori_Case1}, where the problem is expressed as 
  \begin{subequations}\label{Eq:PF_ori_Case1_1}
  	\begin{align}
  		&\label{min PEB} \underset{ {\mathbf{w}}  } {\mathop{\min }}  \quad  \frac{ {{ { [ {\tilde{{\mathbf{T}}}} ]}_{(1,1)} }} } {\Re [{\rm{Tr}}( {\mathbf{w}}  {\dot{{\mathbf{A}}}_{1,{{\bm{\phi}_s}}}^{{\rm{H}}} { {{\dot{{\mathbf{A}}}}}_{ 1, {{\bm{\phi}_s}}} {\mathbf{w}}^{{\rm{H}}}  ) } }  ] }  \\  
  		& {\rm{s.t.}} \quad  \; \eqref{Communication QoS_NOMA_CU}, \eqref{Power max}. 
  	\end{align}
  \end{subequations}

\begin{lemma}\label{Eq_Norm}
    Let ${\mathbf{h}}_{s}={\sqrt{ \left\| {{{\mathbf{v}}_{r,1}}} \right\|_2^2 }} {\bm{{\alpha}}}_{t,1}+ {{\bm{\alpha}}_{t,1}} \odot {{\mathbf{v}}_{t,1}} \in {\mathbb{C}}^{{N}_{t} \times 1}$, we have ${\rm{Tr}}({{\mathbf{R}}_{x}}{\dot{{\mathbf{A}}}}_{1,{\phi_s}}^{\rm{H}} {\dot{{\mathbf{A}}}}_{1, {\phi _s}} ) \approx \frac{{4{\pi ^2}{N_r}}}{{\lambda _c^2}} {  \left| { {{\mathbf{h}}_{s}^{\rm{H}}} {{\mathbf{w}}} } \right|^2} $. 
\end{lemma}
\begin{proof}
    With the definition of ${\mathbf{v}}_{t,1}$, we have ${ [{{\bm{\alpha}}_{t,1}} \odot {{\mathbf{v}}_{t,1}} ]^{\rm{H}} {{\bm{\alpha}}_{t,1}}= \left\|{ {{\mathbf{v}}_{t,1}} }\right\|_{1}  } \approx 0$. Therefore, ${\mathbf{h}}_{s} {\mathbf{h}}_{s}^{{\rm{H}}}=\left\| {{{\mathbf{v}}_{r,1}}} \right\|_2^2 {{\bar{{\mathbf{A}}}}_t}   + {\rm{diag}}({{\mathbf{v}}_{t,1}}){{\bar{{\mathbf{A}}}}_t} {\rm{diag}}({{\mathbf{v}}_{t,1}})=\frac{ {\lambda}_{c}^{2} }{4{\pi}^{2}{N}_{r} } {\dot{{\mathbf{A}}}}_{1,{\phi_s}}^{\rm{H}} {\dot{{\mathbf{A}}}}_{1,{\phi_s}}$. The proof is complete.
\end{proof}

\begin{remark}\label{Lemma2_R1}
    Based on \textbf{Lemma~\ref{Eq_Norm}}, it's clear that $ {\mathop{\min }} \;{\rm{CRB}}_{{\phi
    }}$ is equivalent to ${\mathop{\max }} \; {  \left| { {{\mathbf{h}}_{s}^{\rm{H}}}{{\mathbf{w}}} } \right|^2} $. 
    Similarly, there is a vector ${\mathbf{h}}_{s}^{'} \in {\rm{span}} \left\{ { {{\bm{\alpha}}_t}, {{\bm{\alpha}}_t} \odot {{\mathbf{v}}_{t,2}} } \right\}$ which satisfies ${\rm{Tr}}( {{\mathbf{R}}_{x}}{\dot{{\mathbf{A}}}}_{{\rho_s}}^{\rm{H}} {\dot{{\mathbf{A}}}}_{{\rho_s}} ) \approx \frac{{4{\pi^2}{N_r}}}{{\lambda _c^2}} {  \left| { {{\mathbf{w}}^{\rm{H}}} {{\mathbf{h}}_{s}^{'}} } \right|^2}  $.
    Meanwhile, $[ {{\bm{\alpha}}_{t,1}} \odot {{\mathbf{v}}_{t,2}} ]^{\rm{H}} {{\bm{\alpha}}_{t,1}} \in {\mathbb{R}}$, which indicates that $ {{\bm{\alpha}}_{t,1}} \odot {{\mathbf{v}}_{t,2}} $ is parallel to ${{\bm{\alpha}}_{t,1}}$. In other words, there is an projected vector ${\tilde{{\bm{\alpha}}}}_{{\rho}_{s}} \in {\mathbb{C}}^{{N}_{t} \times 1}$ of ${\dot{{\mathbf{A}}}}_{{\rho_s}}^{\rm{H}} {\dot{{\mathbf{A}}}}_{{\rho_s}} $ which satisfied ${\tilde{{\bm{\alpha}}}}_{{\rho}_{s}}^{\rm{H}} {\tilde{{\bm{\alpha}}}}_{{\rho}_{s}} \approx {\dot{{\mathbf{A}}}}_{1,{\rho_s}}^{\rm{H}} {\dot{{\mathbf{A}}}}_{1,{\rho_s}}$, where ${\tilde{{\bm{\alpha}}}}_{{\rho}_{s}} $ can be calculated as
    \begin{equation}\label{Eq_alpha:rho_s}
       {\tilde{{\bm{\alpha}}}}_{{\rho}_{s}}= \frac{{ {\bm{\alpha}}_{t,1}^{{\rm{H}}} {\dot{{\mathbf{A}}}}_{1,{\rho_s}}^{\rm{H}} {\dot{{\mathbf{A}}}}_{1,{\rho_s}} {\bm{\alpha}}_{t,1} } } {{ {\bm{\alpha}}_{t,1}^{\rm{H}} {\bm{\alpha}}_{t,1} }} \odot {\bm{\alpha}}_{t,1}
    \end{equation}
    Thus, ${\mathop{\min }} \;{\rm{CRB}}_{{\rho
    }_{s}} $ can be transformed into ${\mathop{\max }} \; {  \left| { { {\tilde{{\bm{\alpha}}}}_{\rho}^{\rm{H}}}{{\mathbf{w}}} } \right|^2} $. Furthermore, the ${\tilde{\bm{\alpha}}}_{\rho} $ is not the affine combination of ${\mathbf{h}}_{s} $, which indicates there is a tradeoff between $\min {\rm{CRB}}_{{\phi}_{s}} $ and $\min {\rm{CRB}}_{{\rho}_{s}}$.     
\end{remark}

Problem~\eqref{Eq:PF_ori} under the coplanar case can be transformed into:
    \begin{subequations}\label{Eq:PF_Case1_1_sim} 
        \begin{align}
            &\label{max norm} \underset{ {{\mathbf{w}}} } {\mathop{\max }}  \quad {\left\| {{\mathbf{h}}_{s}^{\rm{H}}} {{\mathbf{w}}} \right\|^2}   \\    
  		& {\rm{s.t.}}\; \; \eqref{Communication QoS_NOMA_CU}, \eqref{Power max},
        \end{align}
    \end{subequations}  
where the problem~\eqref{Eq:PF_Case1_1_sim} reaches the same optimal solution as the problem~\eqref{Eq:PF_ori_Case1_1}.

\begin{proposition}\label{Pro:Opt Solution}
    The optimal solution ${\mathbf{w}}^{*}$ of \eqref{Eq:PF_Case1_1_sim} can be expressed as
    \begin{equation}
    {{\mathbf{w}}^*} = \left\{ \begin{array}{l}
\frac{{\sqrt {{P_{\max }}} {{\mathbf{h}}_s}}}{{\left\| {{{\mathbf{h}}_s}} \right\|}}, \; {\rm{if}} \; {\left| {{\mathbf{h}}_c^{\rm{H}} {{\mathbf{h}}_s}} \right|^2} > \frac{{{{\left\| {{{\mathbf{h}}_s}} \right\|}^2}{\Gamma_c}{\sigma ^2}}}{{{P_{\max}}}} \\
{x_1}{{\mathbf{u}}_s} + {x_2}{{\mathbf{a}}_s}, {\rm{others}}
\end{array} \right.
    \end{equation}where 
    \begin{subequations}\label{Opt_expression_0}
        \begin{equation}\label{Opt_expression_1}
            {{\mathbf{u}}_s} = \frac{{{{\mathbf{h}}_s}}}{{\left\| {{{\mathbf{h}}_s}} \right\|}}, \quad {{\mathbf{a}}_s} = \frac{{{{\mathbf{h}}_c} - ({\mathbf{u}}_{s}^{\rm{H}}{{\mathbf{h}}_s}){{\mathbf{u}}_s}}} {{\left\| {{\mathbf{h}}_c} - ( {\mathbf{u}}_{s}^{\rm{H}}{\mathbf{h}}_s){{\mathbf{u}}_s} \right\|}}, 
        \end{equation}
        \begin{equation}\label{Opt_expression_2}
            {{x_1}} ={\sqrt{{P_{\max }} - \frac{{{\Gamma _c}{\sigma ^2}}}{{{{\left\| {{{\mathbf{h}}_c}} \right\|}^2}}}}} \frac{ { {\mathbf{a}}_{s}^{\rm{H}} {\mathbf{h}}_{s} } } {{\left| {\mathbf{a}}_{s}^{\rm{H}} {\mathbf{h}}_{s} \right|}} , \quad {{x_2}} =  {\sqrt{ \frac{{{\Gamma _c}{\sigma ^2}}}{{{{\left\| {{{\mathbf{h}}_c}} \right\|}^2}}} }} \frac{ {\mathbf{u}}_{s}^{\rm{H}} {\mathbf{h}}_{s} } {{\left| {\mathbf{u}}_{s}^{\rm{H}} {\mathbf{h}}_{s} \right|}} .  
        \end{equation}
    \end{subequations}
\end{proposition}
\begin{proof}
    Please see the details in~\cite[Theorem 1]{FanLiu_22TSP}. 
\end{proof}

\begin{remark}
    From \textbf{Proposition~\ref{Pro:Opt Solution}}, it can be seen that when ${\left| {{\mathbf{h}}_c^{\rm{H}} {{\mathbf{h}}_s}} \right|^2} > \frac{{{{\left\| {{{\mathbf{h}}_s}} \right\|}^2}{\Gamma_c}{\sigma ^2}}}{{{P_{\max}}}} $, the radio sensing functionality plays a dominant role where the radio performance is close to the sensing-only scheme. When ${\left| {{\mathbf{h}}_c^{\rm{H}} {{\mathbf{h}}_s}} \right|^2} \leq \frac{{{{\left\| {{{\mathbf{h}}_s}} \right\|}^2}{\Gamma_c}{\sigma ^2}}}{{{P_{\max}}}}$, the optimal solution is composed by ${\mathbf{u}}_{s} $ and ${\mathbf{a}}_{s}$. Meanwhile, the projected vector ${\mathbf{h}}_{s}$ substituted the previous ${\bm{\alpha}}_{s} $ in the FF propagation, which indicates that the radio sensing dominant region is extended in NF propagation. This is because the distance domain resolution is introduced under NF propagation, which further increases the spatial DoF.
\end{remark}

\subsubsection{BF Design to minimize the SPEB} 
Inspired by the optimal solution of \eqref{Eq:PF_Case1_1_sim}, the denominator term of the ratio in the objective function can be converted into the 2-norm of an inner product between ${\mathbf{w}}$ and the projected vector ${\tilde{{\bm{\alpha}}}}_{i}, i \in \left\{ {\rho}_{s}, {\phi}_{s},{y}_{s} \right\}$. Besides, it is noted that ${\rm{CRB}}_{{\rho}_{s}}$ can be expressed as the inner product between the BF vector ${\mathbf{w}}$ and the project vector ${\tilde{{\bm{\alpha}}}}_{{\rho}_{s}}$, the objective function in \eqref{Eq:PF_ori_Case1} can be derived as
\begin{equation}\label{2FP}
    {\rm{OF}} ( {\mathbf{w}} ) = \frac{ {{ { [ {\tilde{{\mathbf{T}}}} ]}_{(1,1)} }} }{{{{\left| {{\mathbf{h}}_s^{{\rm{H}}} {\mathbf{w}} } \right|}^2}}} + \frac{ {{ { [ {\tilde{{\mathbf{T}}}} ]}_{(2,2)} }} }{{{{\left| {\tilde{{{\bm{\alpha}}}}}_{{\rho}_{s}}^{{\rm{H}}} {\mathbf{w}} \right|}^2}}}.
\end{equation}

\noindent As ${\mathbf{h}}_{s}$ is not the linear combination of ${\tilde{\bm{\alpha}}}_{{\rho}_{s}}$, \eqref{2FP} has non-compatible fractional terms, where the problem \ref{Eq:PF_ori_Case1} turns into multi-ratio fractional programming (FP). The conventional Dinkelbach methods cannot guarantee convergence among the consecutive iterations as the optimal solution of each ratio is not equal. Meanwhile, the conventional SDR approach introduces the auxiliary matrix variables ${\mathbf{W}}={\mathbf{w}}{\mathbf{w}}^{{\rm{H}}} $ to deal with the quadratic form, which increased the computation complexity. Motivated by these points, we employ the quadratic transformation~\cite{QuadraticTransformation} on \eqref{2FP} in the following. 

\begin{proposition}\label{Pro:QF_equivalent}
  The sum-of-ratios minimization problem \eqref{Eq:PF_ori_Case1} is equivalent to 
    \begin{subequations}\label{Eq:PF_C1_QF}
        \begin{align}
            & \label{max norm} \underset{ {\mathbf{w}}, {\tilde{{\mathbf{y}}}} } {\min } \;  {\tilde{{\rm{OF}}}} ({\mathbf{w}}, {\tilde{{\mathbf{y}}}}) \\    
  		& \label{Eq:PF_C1_QF:SINR} {\rm{s.t.}}\; \; {\mathbf{h}}_{c}^{{\rm{H}}} {\mathbf{w}} \ge \sqrt {{\sigma ^2}{{\bar \gamma }_{\min }}} ,  \\
            & \; \; \quad \eqref{Power max},
        \end{align}
    \end{subequations}where ${\tilde{{\rm{OF}}}} ({\mathbf{w}}, {\tilde{{\mathbf{y}}}}) =\frac{1}{{{{\left[ {2\tilde y_1 \sqrt {{{\left|{ {\mathbf{h}}_{s}^{{\rm{H}}} {\mathbf{w}} } \right|}^2}}  - {\tilde{y}}_{1}^{2} { [ {\tilde{{\mathbf{T}}}} ]}_{(1,1)}  } \right]}_ + }}} + \frac{1}{{{{\left[ {2\tilde y_2 \sqrt {{{\left| { {\tilde{{\bm{\alpha}}}}_{{\rho}_{s}}^{{\rm{H}}} {\mathbf{w}} } \right|}^2}}  - {\tilde{y}}_2^{2} { [ {\tilde{{\mathbf{T}}}} ]}_{(2,2)} } \right]}_ + }}}$, and ${\tilde{{\mathbf{y}}}}=\left[ { {\tilde{y}}_{1}, {\tilde{y}}_{2} } \right] $ denotes auxiliary variable. 
    \end{proposition}

\begin{proof}
    When ${\mathbf{w}}$ is fixed, the optimal ${\tilde{y}}_{n}$ is determined via ${\tilde{y}}_{1}^{*}=\frac{{\sqrt {{{\left| {{\mathbf{h}}_{s}^{{\rm{H}}} {\mathbf{w}} } \right|}^2}} }} { { [ {\tilde{{\mathbf{T}}}} ]}_{(1,1)} }, {\tilde{y}}_{2}^{*}=\frac{{\sqrt {{{ \left| { {\tilde{{\bm{\alpha}}}}_{{\rho}_{s}}^{{\rm{H}}} {\mathbf{w}} } \right|}^2}} }} { { [ {\tilde{{\mathbf{T}}}} ]}_{(2,2)} } $. By substituting the above equation into the objective function in \eqref{Eq:PF_C1_QF}, the same optimal value is reached between the problem \eqref{Eq:PF_ori_Case1} and \eqref{Eq:PF_C1_QF}.
\end{proof}
With the fixed ${\tilde{{\mathbf{y}}}}$ in each iteration, the problem \eqref{Eq:PF_ori_Case1} can be transformed into 
\begin{subequations}\label{Eq:PF_C1_QF_1}
        \begin{align}
            & \label{max norm} \underset{ {\mathbf{w}} } {\min } \; {\tilde{{\rm{OF}}}} ({\mathbf{w}}, {\tilde{{\mathbf{y}}}})  \\    
  		& {\rm{s.t.}}\; \; \eqref{Power max}, \eqref{Eq:PF_C1_QF:SINR}, 
        \end{align}
\end{subequations}which is a convex problem and the optimal solution ${\mathbf{w}}^{*}_{i}$ can be obtained via CVX~\cite{cvx}. Once the optimal solution within the iteration is derived, the auxiliary variable ${\tilde{{\mathbf{y}}}} $ is updated via 
\begin{equation}\label{y_Case1_Update}
    {\tilde{y}}_{1}^{(i+1)}= \frac{{\sqrt {{{\left| {{\mathbf{h}}_{s}^{{\rm{H}}} {\mathbf{w}}^{*}_{i} } \right|}^2}} }} { { [ {\tilde{{\mathbf{T}}}} ]}_{(1,1)} },  {\tilde{y}}_{2}^{(i+1)}= \frac{{\sqrt {{{\left| { {\tilde{{\bm{\alpha}}}}_{{\rho}_{s}}^{{\rm{H}}} {\mathbf{w}}^{*}_{i} } \right|}^2}} }} { { [ {\tilde{{\mathbf{T}}}} ]}_{(2,2)} }. 
\end{equation}The details are summarised in {\textbf{Algorithm~\ref{Alg_VQF_Case1}}}.

\subsubsection{Convergence and Complexity Analysis} 

The convergence of \textbf{Algorithm~\ref{Alg_VQF_Case1}} is summarized in the following. With the fixed auxiliary variable ${\tilde{\mathbf{y}}}^{(i)} $ in each iteration, the optimal solution of \eqref{Eq:PF_C1_QF} can be reached. Among the consecutive iterations, the objective function is non-increasing. Thus, {\textbf{Algorithm~\ref{Alg_VQF_Case1}}} converges to a stable suboptimal solution. For the complexity analysis, the convex problem of \eqref{Eq:PF_C1_QF} can be solved via the inter-point methods with the complexity of ${\mathcal{O}} (4{N}_{t}^{2}) $~\cite{InteriorPoint_Complexity_0}. By denoting the iteration number as $I_{1}$, the complexity of {\textbf{Algorithm~\ref{Alg_VQF_Case1}}} can be expressed as ${\mathcal{O}} { (4{I}_{1}{N}_{t}^{2})}$. 

 \begin{algorithm}[tp]
      \caption{Proposed vector-based quadratic transformation (VQF) algorithm for \eqref{Eq:PF_ori_Case1}}
    	\label{Alg_VQF_Case1}
    		\begin{algorithmic}[1]\label{alg_1} 
    			\STATE Initialize feasible points $ {\mathbf{{w}}}^{0}= {{\mathbf{w}^{*}_{\phi }}}$ and ${\tilde{{\bm{y}}}}_{s}^{0}$ is initialized by \eqref{y_Case1_Update} with ${\mathbf{{w}}}^{0}$. 		
    			\REPEAT 
    			\STATE {Update $ {{\mathbf{w}}^{(i)}}$ by solving \eqref{Eq:PF_C1_QF}.} 
    			\STATE  Update ${\tilde{{\mathbf{y}}}}^{i+1}$ with $ {{\mathbf{w}}^{(i)}}$ by \eqref{y_Case1_Update}.  
    			\UNTIL {The relevant variation of objective value for \eqref{Eq:PF_C1_QF} among consecutive iteration is below a predefined threshold ${\epsilon}_{con} $}.
    			\STATE \textbf{Output:$ {{\mathbf{w}}^{*}},{\tilde{{\mathbf{y}}}}$}. 
    		\end{algorithmic}
 \end{algorithm}

\section{Radio Sensing Performance Analysis and BF Design for Non-coplanar Case}

\subsection{Radio Sensing Performance Analysis with Non-coplanar Case}
\subsubsection{CRB Expression with USW Channel Model} 
By denoting ${\left[ {{{\mathbf{v}}_{l,2,1}}} \right]_m} = [ \frac{{{\rho _s}}} {{\sqrt {\rho _s^2 + y_s^2} }} - \frac{{{\rho _s} - {R_{{\rm{CA}}}\cos ({\phi _s} - \phi _s^{(m)})}}}{{\sqrt {{{(\rho _{{\rm{CA}}}^{(m)})}^2} + y_s^2} }} ]$, ${\left[ {{{\mathbf{v}}_{l,2,2}}} \right]_m} = \frac{{\sin ({\phi _s} - \phi _s^{(m)})}}{{\sqrt {{{(\rho _{{\rm{CA}}}^{(m)})}^2} + y_s^2} }}$, and ${\left[ {{{\mathbf{v}}_{l,2,3}}} \right]_m} = \frac{1}{{\sqrt {\rho _{\rm{s}}^2 + y_s^2} }} - \frac{1}{{\sqrt {{{(\rho _{{\rm{CA}}}^{(m)})}^2} + y_s^2} }}$, the corresponding derivatives ${\dot {\bm{\alpha}}_{{l,2,i}}}, l \in \left\{r,t\right\}, i=1,2,3 $ of $\left\{{\rho}_{s}, {\phi}_{s},{y}_{s} \right\}$ are given by
    \begin{subequations}\label{Case2_deriva_0}
    	\begin{equation}
    		{\dot{{\bm{\alpha}}}}_{l,2,1}= \frac{{\partial {\bm{\alpha}}_{l,2}}} {{\partial {\rho_s}}} = \frac{{j2\pi }}{{{\lambda _c}}}{{\mathbf{v}}_{l,2,1}} \odot {{\bm{\alpha}}_{l,2}}, 
    	\end{equation}
        \begin{equation}
        	{\dot{{\bm{\alpha}}}}_{l,2,2}= \frac{{\partial {{\bm{\alpha}}}_{l,2}}}{{\partial {\phi_s}}} = \frac{{ - j2\pi {R_{{\rm{CA}}} {\rho _s}}}}{{{\lambda _c}}}{ {\mathbf{v}}_{l,2,2}} \odot { {\bm{\alpha}}_{l,2}}, 
        \end{equation}
        \begin{equation}
        	{\dot{{\bm{\alpha}}}} _{l,2,3} = \frac{{\partial {\bm{\alpha}}_{l,2}}} {{\partial {y_s}}} = \frac{{j2\pi {y_s}}}{{{\lambda _c}}} {\mathbf{v}}_{l,2,3} \odot {\bm{\alpha}}_{l,2},
        \end{equation}
    \end{subequations}

    \begin{lemma}\label{Lemma_Case2_norm}
       Due to the periodical property of the function in the auxiliary vectors $\left\{ {\mathbf{v}}_{l,2,1}, {\mathbf{v}}_{l,2,2}, {\mathbf{v}}_{l,2,3}\right\}$, the norm properties of the auxiliary vectors are given by \eqref{Auxiliary_Periodical}, where ${\gamma}_{1}=2{R}_{{\rm{CA}}} {\rho}_{s} $, ${\gamma}_{2}={y}_{s}^{2}+{\rho}_{s}^{2}+{R}_{\rm{CA}}^{2} $, ${\mathbf{K}} \left( \alpha \right)=\int_0^{\pi /2} {\frac{1}{{\sqrt {1 - {{\alpha}^2}{{\sin }^2}\theta } }}} {\kern 1pt} d\theta $ denotes the complete elliptic integral of the first kind~\cite{Integral_reference}.

      \begin{figure*}
            \begin{small}
                \begin{subequations}\label{Auxiliary_Periodical}
             \begin{equation}\label{Auxiliary_Periodical_1}
               \begin{split}
                  \frac{\left\| {{{\mathbf{v}}_{l,2,1}}} \right\|_1 }{ {N}_{l}} =\frac{{{\rho _s}}}{{\sqrt {\rho _s^2 + y_s^2} }} - \frac{{2{\rho _s}}}{{\pi \sqrt {{\gamma _2} + {\gamma _1}} }}K\left( {\sqrt {\frac{{2{\gamma _1}}}{{{\gamma _2} + {\gamma _1}}}} } \right) 
                   + \frac{ { {R_{{\rm{CA}}}} \kappa {\rm{(}}{\gamma _2}{\rm{,}}{\gamma _1}{\rm{)}}}} {\pi } , \quad   {\left\| {{\mathbf{v}}_{l,2,3}} \right\|_1}  = \frac{{{N_{l}}}}{{\sqrt {\rho_{\rm{s}}^2 + y_s^2} }} - \frac{{2{N_{l}}}} {{\pi \sqrt {{\gamma _2} + {\gamma_1}} }} {\mathbf{K}}\left( {\sqrt {\frac{{2{\gamma _1}}}{{{\gamma _2} + {\gamma _1}}}} } \right),
               \end{split}
            \end{equation}
            \begin{equation}\label{Case2_norm2_1}
            \begin{split}
            \frac{\left\| {{{\mathbf{v}}_{l,2,1}}} \right\|_2^2 }{ {N}_{l} } = \frac{{\rho _s^2}}{{\rho _s^2 + y_s^2}}-\frac{{2{\rho _s}}}{{\sqrt {\rho _s^2 + y_s^2} }}[\frac{ {2{\rho _s}} K\left( {\sqrt {\frac{{2{\gamma _1}}}{{{\gamma _2} + {\gamma _1}}}} } \right) } {{\pi \sqrt {{\gamma _2} + {\gamma _1}} }} - \frac{{ {R_{{\rm{CA}}}} \kappa {\rm{(}}{\gamma _2}{\rm{,}}{\gamma _1}{\rm{)}}}} {\pi }]  + \frac{{\{ 2\rho _s^2 - 2({\gamma _2} - \sqrt {\gamma _2^2 - \gamma _1^2} ) + R_{{\rm{CA}}}^2 + \frac{{2\gamma _2^2 - \gamma _1^2 - 2{\gamma _2}\sqrt {\gamma _2^2 - \gamma _1^2} }}{{4\rho _s^2}}\} }}{{2\sqrt {\gamma _2^2 - \gamma _1^2} }}, 
            \end{split}
            \end{equation}
            \begin{equation}\label{Case2_norm2_3}
                \left\| {{{\mathbf{v}}_{l,2,2}}} \right\|_2^2=\frac{{{N_{l}}}}{{\gamma_1^2}}\left( {\gamma _2} - \sqrt {\gamma _2^2 - \gamma _1^2} \right), \quad   \frac{\left\| {{{\mathbf{v}}_{l,2,3}}} \right\|_2^2}{{N}_{l}} = \frac{1}{{\rho_{s}^2 + y_s^2}} + \frac{1}{{\sqrt {\gamma_2^2 - \gamma _1^2} }}-\frac{4  {\mathbf{K}}\left( {\sqrt {\frac{{2{\gamma _1}}}{{{\gamma _2} + {\gamma _1}}}} } \right) } {{\pi \sqrt {\rho_s^2 + y_s^2} \sqrt {{\gamma_2} + {\gamma _1}} }}
            \end{equation}
            \begin{equation}\label{Auxiliary_Periodical_0}
    			\left\| {{\mathbf{v}}_{l,2,2}} \right\|_1  =\left\| {{{\mathbf{v}}_{l,2,1}} \odot {{\mathbf{v}}_{l,2,2}}} \right\|_2^1=\left\| {{{\mathbf{v}}_{l,2,2}} \odot {\mathbf{v}}_{l,2,3}} \right\|_2^1= 0, 
    		\end{equation}
            \hrulefill
    	\end{subequations}
            \end{small}
       \end{figure*}
\noindent 

\end{lemma}

    \begin{proof}
     Please see the details in {\textbf{Appendix~\ref{App:Lemma_Case2_norm}}}. 
    \end{proof}

    \begin{proposition}\label{Pro_FIM_Case2}
       Based on the USW model of the NF propagation, the FIM $ {\mathbf{J}}_{{\bm{\gamma}} }$ under non-coplanar case can be approximated as a block diagonal matrix ${\mathbf{J}}_{ {\bm{\gamma}} }^{{\rm{bd}}}$ with trinity loss
       \begin{equation}\label{FIM_bd}
        	\begin{split}
        		{\mathbf{J}}_{ {\bm{\gamma}} } \approx {\mathbf{J}}_{ {\bm{\gamma}} }^{{\rm{bd}}} &=  [ \begin{array}{*{20}{c}}
        			{\begin{array}{*{30}{c}}
        					\left[ {\mathbf{J}}_{ {\bm{\gamma}} } \right]_{{\rm{11}}} \\
        					0 \\
        					0
        			\end{array}}& {\begin{array}{*{30}{c}}
        					0 \\
        					\left[ {\mathbf{J}}_{ {\bm{\gamma}} } \right]_{{\rm{22}}}  \\
        					\left[ {\mathbf{J}}_{ {\bm{\gamma}} } \right]_{{\rm{23}}}^{{\rm{H}}}  
        			\end{array}}
        			& {\begin{array}{*{30}{c}}
        					0 \\
        					\left[ {\mathbf{J}}_{ {\bm{\gamma}} } \right]_{{\rm{23}}}   \\
        					\left[ {\mathbf{J}}_{ {\bm{\gamma}} } \right]_{{\rm{33}}}
        			\end{array}}
        		\end{array} ]. 
        	\end{split}
        \end{equation}The corresponding separate CRBs over different coordinates are given by 

        \begin{small}
            \begin{subequations}\label{General_CRB_Ori}
     		\begin{equation}
     			{\rm{CRB}}_{{\phi}_{s}} = [ {\mathbf{J}}_{ {\mathbf{p}}_{s} }^{-1} ]_{{\rm{11}}} = \frac{1}{{{{\left[ {{{\bf{J}}_{{{\bf{p}}_s}}}} \right]}_{{\rm{11}}}}}} - \frac{ { {{ \left[ {{{\bf{J}}_{{{\bf{p}}_s}}}} \right]}_{{\rm{23}}}} \left[ {{{\bf{J}}_{{{\bf{p}}_s}}}} \right]_{{\rm{23}}}^{\rm{H}}}} {{{{\left[ {{{\bf{J}}_{{{\bf{p}}_s}}}} \right]}_{{\rm{11}}}} {{\left[ {{{\bf{J}}_{{{\bf{p}}_s}}}} \right]}_{{\rm{22}}}} {{\left[ {{{\bf{J}}_{{{\bf{p}}_s}}}} \right]}_{{\rm{33}}}}}}, 
     		\end{equation}
     	    \begin{equation}
     	    	{\rm{CRB}}_{{\rho}_{s}} \approx [ {\mathbf{J}}_{ {\mathbf{p}}_{s} }^{-1} ]_{{\rm{22}}} = \frac{{{\sigma ^2}}} {{2{{\left| {{\alpha_s}} \right|}^2}L\Re  \{  {\rm{Tr}} [ { {\mathbf{R}}_{x} } \mathop {\dot{{\mathbf{A}}}}_{2,{\rho_s}}^{\rm{H}} {{\mathop {\dot{{\mathbf{A}}}} }_{2,{\rho_s}} } ] \} }}, 
     	    \end{equation}
            \begin{equation}
            	{\rm{CRB}}_{y_{s}} \approx [ {\mathbf{J}}_{ {\mathbf{p}}_{s} }^{-1} ]_{{\rm{33}}} = \frac{{{\sigma ^2}}} {{2{{\left| {{\alpha _s}} \right|}^2}L\Re  \{  {\rm{Tr}} [{\mathbf{R}}_{x} \mathop {\dot{{\mathbf{A}}}}_{2,{y_s}}^{\rm{H}} {{\mathop {\dot{{\mathbf{A}}}} }_{2,{y_s}} } ] \} }},
            \end{equation}
     	\end{subequations}
        \end{small}
       
       \noindent where $\left\{ {\dot{{\mathbf{A}}}}_{2,{\phi_s}}^{\rm{H}} {{ {\dot{{\mathbf{A}}}} }_{2,{\phi_s}} }, {\dot{{\mathbf{A}}}}_{2,{\rho_s}}^{\rm{H}} {{ {\dot{{\mathbf{A}}}} }_{2,{\rho_s}} }, {\dot{{\mathbf{A}}}}_{2,{y_s}}^{\rm{H}} {{ {\dot{{\mathbf{A}}}} }_{2,{y_s}} } \right\} $ are given by \eqref{General_AA_Diagonal}. With the isotropic BF design at the BS ${\mathbf{R}}_{x}= \frac{P_{\max} }{{N}_{t}} {\mathbf{I}}_{{N}_{t}\times {N}_{t}}$, the corresponding CRBs of $\left\{ {\rho}_{s}, {y}_{s} \right\}$ are given by \eqref{General_CRB_Iso_USW}.
        \begin{subequations}\label{General_CRB_Iso_USW}
            \begin{equation}
                {\rm{CRB}}_{\rho_{s}}= \frac{{{\sigma ^2}\lambda _c^2}}{{8{{\left| {{\alpha _s}} \right|}^2}{\pi ^2}{P_{\max }}{N_r}L {\Phi}_{{\rho}_{s}} ({\phi _s}{\rm{,}}{\rho _s}, {y_s} ) }} 
            \end{equation}
            \begin{equation}
                {\rm{CRB}}_{y_{s}}=\frac{{{\sigma ^2}\lambda _c^2}}{{16{{\left| {{\alpha _s}} \right|}^2}{\pi ^2}y_s^2{P_{\max }}{N_r}L{\Phi}_{{y}_{s}} {\rm{(}}{\phi _s}{\rm{,}}{\rho _s}{\rm{,}}{{\rm{y}}_s}{\rm{)}} } },
            \end{equation}
        \end{subequations}where ${\phi}_{{\rho}_{s}} {\rm{(}}{\phi _s}{\rm{,}}{\rho _s}{\rm{,}}{{\rm{y}}_s}{\rm{)}} = \frac{{{{\left\| {{{\mathbf{v}}_{r,2,1}}} \right\|}_2}}}{{{N_r}}} - {\frac{{{{\left\| {{{\mathbf{v}}_{r,2,1}}} \right\|}_1^2}}}{ {N_r^2} }} $ and ${\Phi}_{{y}_{s}} {\rm{(}}{\phi _s}{\rm{,}}{\rho _s}{\rm{,}}{{\rm{y}}_s}{\rm{)}} =\frac{1}{{\sqrt {\gamma _2^2 - \gamma _1^2} }} - \frac{4 {K^2}\left( {\sqrt {\frac{{2{\gamma _1}}}{{{\gamma _2} + {\gamma _1}}}} } \right) } {{{\pi ^2}({\gamma _2} + {\gamma _1})}} $. 
        \end{proposition}
        \begin{proof}
            Please see the details in {\textbf{Appendix~\ref{App:Pro_FIM_Case2}}}. 
    \end{proof}

\subsection{Problem Formulation and BF Design for Non-coplanar Case}
For non-coplanar case, the ST is located beyond the coplanar of the UCA under the NF propagation. The problem~\eqref{Eq:PF_ori} can be recast as
\begin{subequations}\label{Eq:PF_ori_Case2}
  	\begin{align}
  		&\label{PF_C2:min PEB} \underset{ {\mathbf{w}}  } {\mathop{\min }}  \quad  {\rm{SPEB}}_{2} \left( {\mathbf{R}}_{x} ; {\bm{\eta}} \right) \\  
  		\label{PF_C2:Communication QoS} & {\rm{s.t.}} \quad  \; {R}_{c} \ge {\bar{R}}_{\min}, \\
  		\label{PF_C2:Power max}& \quad \quad \quad   {\mathbf{w}}^{{\rm{H}}} {\mathbf{w}}\le  P_{\max}.
 	\end{align}
\end{subequations}where ${\rm{SPEB}}_{2} \left( {\mathbf{R}}_{x} ; {\bm{\eta}} \right)= {\rm{tr}} { ( {{\mathbf{T}}}^{{\rm{H}}} {\mathbf{J}}_{ {\bm{\gamma}} }^{{\rm{bd}}} {{\mathbf{T}}})}^{-1}   $. The problem~\eqref{Eq:PF_ori_Case2} is a non-convex programming problem due to the quadratic form. By introducing the auxiliary variables ${\mathbf{W}}={\mathbf{w}}^{{\rm{H}}} {\mathbf{w}} \in {\mathbb{C}}^{ {N}_{t} \times {N}_{t} }$, the problem \eqref{Eq:PF_ori_Case2} can be converted into
\begin{subequations}\label{Eq:PF_SDR_Case2}
    \begin{align}
  	&\label{PF_C2:min PEB} \underset{ {\mathbf{W}} } {\mathop{\min }}  \quad  {\rm{SPEB}}_{2} \left( {\mathbf{W}} ; {\bm{\eta}} \right) \\  
  		\label{PF_C2:Communication QoS} & {\rm{s.t.}} \quad  \;  {\rm{tr}} ({\mathbf{H}} {\mathbf{W}} )   \ge {\sigma}^{2} {\bar{\gamma}}_{\min}, \\
  		\label{PF_C2:Power max} & \quad  \quad  \quad  {\rm{tr}}( {\mathbf{W}}) \le  P_{\max}, \\
            \label{PF_C2_SDR: rank} &  \quad \quad  \quad {\rm{rank}} ( {\mathbf{W}} )=1. 
 	\end{align}
\end{subequations}As ${\rm{SPEB}}_{2} \left( {\mathbf{W}} ; {\bm{\eta}} \right) $ is linear w.r.t. ${\mathbf{W}} $, problem \eqref{Eq:PF_SDR_Case2} is semi-definite programming, which can be solved via the SDR. The rank-one solution of problem \eqref{Eq:PF_SDR_Case2} can be reconstructed by the Gaussian randomization~\cite{SDR_1}. The complexity of SDR is at about ${\mathcal{O}}({N}_{t}^{4.5})$, which is a high burden with the large transmit antennas. Inspired by the problem \eqref{Eq:PF_C1_QF} under coplanar case, we further propose a low-complexity extended VQF algorithm based on the approximated diagonal matrix ${\mathbf{J}}_{{\bm{\gamma}}}^{{\rm{d}}} $, where the diagonal FIM ${\mathbf{J}}_{{\bm{\gamma}}}^{{\rm{d}}}$ under non-coplanar case is expressed as
\begin{equation}\label{FIM_bd}
       \begin{split}
          {\mathbf{J}}_{ {\bm{\gamma}} }^{{\rm{d}}} = {\rm{Diag}} {\rm{\{ [}}{{\mathbf{J}}_{ {\bm{\gamma}} }}{{\rm{]}}_{11}},{{\rm{[}}{{\mathbf{J}}_{ {\bm{\gamma}} }}{\rm{]}}_{22}},{{\rm{[}}{{\bf{J}}_{ {\bm{\gamma}} }}{\rm{]}}_{33}} \}. 
       \end{split}
\end{equation}
Similar to \textbf{Lemma~\ref{Eq_Norm}} and \textbf{Remark~\ref{Lemma2_R1}}, the Rayleigh quotient can be utilized to decompose the partial derivative matrix (i.e., ${\dot{{\mathbf{A}}}}_{2,i}^{\rm{H}} {\dot{{\mathbf{A}} } }_{2,i}, \; i \in \{ { {\rho_s}, {\phi_s}, y_{s} } \} $) on its main eigenvector direction as shown in {\textbf{Lemma~\ref{Pro:Case2_decompose}}}. 

\begin{lemma}\label{Pro:Case2_decompose}
    Let ${\bar{{\mathbf{h}}}}_{s}={\sqrt{ \left\| {{{\mathbf{v}}_{r,2,2}}} \right\|_2^2   }} {{\bm{{\alpha}}}}_{t,2}+ {{\bm{{\alpha}}}}_{t,2} \odot {\mathbf{v}}_{t,2,2} \in {\mathbb{C}}^{{N}_{t} \times 1},$ we have ${\rm{Tr}}({{\mathbf{R}}_{x}}{\dot{{\mathbf{A}}}}_{{2, \phi _s}}^{\rm{H}} {\dot{{\mathbf{A}}}}_{{2, \phi_s}} ) \approx \frac{{4{\pi^2}{N_r}}}{{\lambda _c^2}} {  \left| { { {\bar{{\mathbf{h}}}}_{s}^{\rm{H}} } { {\mathbf{w}} } } \right|^2} $. Similar to coplanar case, ${\rm{Tr}}({{\mathbf{R}}_{x}}{\dot{{\mathbf{A}}}}_{{2, \rho_s}}^{\rm{H}} {\dot{{\mathbf{A}}}}_{{2, \rho_s}} ) $ and ${\rm{Tr}}({{\mathbf{R}}_{x}}{\dot{{\mathbf{A}}}}_{{2, y_s}}^{\rm{H}}{\dot{{\mathbf{A}}}}_{{2, y_s}} ) $ can be decomposed by
    
    \begin{small}
         \begin{equation}
            {\rm{Tr}}({{\mathbf{R}}_{x}}{\dot{{\mathbf{A}}}}_{{2, \rho_s}}^{\rm{H}} {\dot{{\mathbf{A}}}}_{{2, \rho_s}} ) \approx {  \left| {  { {\bar{{\bm{\alpha}}}}_{{\rho}_{s} }^{\rm{H}} } {{\mathbf{w}}} } \right|^2},  {\rm{Tr}}({{\mathbf{R}}_{x}}{\dot{{\mathbf{A}}}}_{{2, y_s}}^{\rm{H}} {\dot{{\mathbf{A}}}}_{{2, y_s}} ) \approx {  \left| {  { {\bar{{\bm{\alpha}}}}_{{y}_{s}}^{\rm{H}} } {{\mathbf{w}}} } \right|^2},
        \end{equation}
    \end{small} where 
    \begin{small}
    \begin{subequations}
        \begin{equation}
            {\bar{{\bm{\alpha}}}}_{{\rho}_{s} }= \frac{{ {{\bm{\alpha}}}_{t,2}^{{\rm{H}}} {\dot{{\mathbf{A}}}}_{{2,\rho_s}}^{{\rm{H}}} {\dot{{\mathbf{A}}}}_{{2,\rho_s}} {{\bm{\alpha}}}_{t,2} } } {{ {{\bm{\alpha}} }_{t,2}^{{\rm{H}}} {\bm{\alpha}}_{t,2} }} \odot {{{\bm{\alpha}} }}_{t,2}, 
        \end{equation}
        \begin{equation}
          {\bar{{\bm{\alpha}}}}_{{y}_{s} }= \frac{{ {{\bm{\alpha}}}_{t,2}^{{\rm{H}}} {\dot{{\mathbf{A}}}}_{{2,y_s}}^{{\rm{H}}} {\dot{{\mathbf{A}}}}_{{2,y_s}} {{\bm{\alpha}}}_{t,2} } } {{ {{\bm{\alpha}}}_{t,2}^{{\rm{H}}} {{\bm{\alpha}}}_{t,2} }} \odot {{\bm{\alpha}}}_{t,2}.
        \end{equation}
    \end{subequations}
    \end{small}
    
\end{lemma}

\begin{proof}
    Please see the details in {\textbf{Appendix~\ref{App:Case2_decompose}}}.  
\end{proof}

Thus, the objective function of problem \eqref{Eq:PF_ori_Case2} can be further simplified as: 
\begin{equation}\label{OF_Case2}
    {\rm{OF}}_{2} ( {\mathbf{w}} )= \frac{ {{ { [ {\tilde{{\mathbf{T}}}} ]}_{(1,1)} }} }{{{{\left| { {\bar{{\mathbf{h}} }}_s^{\rm{H}} {\mathbf{w}} } \right|}^2}}} + \frac{ {{ { [ {\tilde{{\mathbf{T}}}} ]}_{(2,2)} }} } { { { {\left| {\bar{{{\bm{\alpha}}}}}_{{\rho}_{s}}^{\rm{H}} {\mathbf{w}} \right|}^2 }}} + \frac{ {{ { [ {\tilde{{\mathbf{T}}}} ]}_{(3,3)} }} }{{{ {\left| {\bar{{{\bm{\alpha}}}}}_{{y}_{s}}^{\rm{H}} {\mathbf{w}} \right|}^2 }}}.
\end{equation}
It is noted that \eqref{OF_Case2} has the same structure as \eqref{Eq:PF_C1_QF}, where \textbf{Proposition~\ref{Pro:QF_equivalent}} also holds for non-coplanar case. To solve the problem ~\eqref{OF_Case2} in a more tractable form, we introducing the auxiliary variable ${\bar{{\mathbf{y}}}}= [ {\bar{y}}_{1},  {\bar{y}}_{2}, {\bar{y}}_{3}  ]^{\rm{T}} $, where ${\bar{y}}_{1}  \buildrel \Delta \over =\frac{{\sqrt {{{\left| {{\mathbf{h}}_{s}^{\rm{H}} {\mathbf{w}} } \right|}^2}} }} { { [ {\tilde{{\mathbf{T}}}} ]}_{(1,1)} } $, ${\bar{y}}_{2} \buildrel \Delta \over = \frac{{\sqrt {{ {\left| {\bar{{{\bm{\alpha}}}}}_{{\rho}_{s}}^{\rm{H}} {\mathbf{w}} \right|}^2 }} }} { { [ {\tilde{{\mathbf{T}}}} ]}_{(2,2)}}$, and ${\bar{y}}_{3} \buildrel \Delta \over = \frac{{\sqrt {{ {\left| {\bar{{{\bm{\alpha}}}}}_{{y}_{s}}^{\rm{H}} {\mathbf{w}} \right|}^2 }} }} { { [ {\tilde{{\mathbf{T}}}} ]}_{(2,2)}}$. The problem can be recast as
\begin{subequations}\label{Eq:PF_C2_QF}
    \begin{align}
            & \label{max norm} \underset{ {\mathbf{w}}, {\bar{{\mathbf{y}}}} } {\min } \; {\bar{{\rm{OF}}}} ({\mathbf{w}}, {\bar{{\mathbf{y}}}} )   \\ 
  		& {\rm{s.t.}}\; \; \eqref{Power max}, \eqref{Eq:PF_C1_QF:SINR},
    \end{align}
\end{subequations}where $ {\bar{{\rm{OF}}}} ({\mathbf{w}}, {\bar{{\mathbf{y}}}})  =\frac{1}{{{{\left[ {2 {\bar{y}}_1 {{\left|{ {\bar{{\mathbf{h}}}}_{s}^{\rm{H}} {\mathbf{w}} } \right|}}  - {\bar{y}}_{1}^{2}  { [ {\tilde{{\mathbf{T}}}} ]}_{(2,2)} } \right]}_ + }}} + \sum\limits_{i = 1,3} {\frac{1}{{{{\left[ {2{{\bar{y}}_i}\left| {\bar{{\bm{\alpha}} }_{i}^{\rm{H}} {\mathbf{w}} } \right| - \bar y_i^2   { [ {\tilde{{\mathbf{T}}}} ]}_{(j,j)}} \right]}_ + }}}} $. The problem \eqref{Eq:PF_C2_QF} is a convex problem and the optimal solution ${\mathbf{w}}_{(i)}^{*}$ can be obtained via CVX~\cite{cvx}. For each iteration, the auxiliary variable ${\bar{{\mathbf{y}} }}$ can be updated by 
\begin{equation}\label{y_Case2_Update}
     {\bar{y}}_{1}^{(i+1)}= \frac{{\sqrt {{{\left| { {\bar{{\mathbf{h}}}}_{s}^{\rm{H}}{\mathbf{w}}^{*}_{i} } \right|}^2}} }} { { [ {\tilde{{\mathbf{T}}}} ]}_{(1,1)} },  {\bar{y}}_{j}^{(i+1)}= \frac{{\sqrt {{{\left| { {\bar{{\bm{\alpha}}}}_{u_l}^{\rm{H}} {\mathbf{w}}^{*}_{i} } \right|}^2}} }} { { [ {\tilde{{\mathbf{T}}}} ]}_{(l,l)} }, l=2,3, 
\end{equation}where ${u}_{2}={\rho}_{s}, {u}_{3}={y}_{s}$.The corresponding algorithm for \eqref{Eq:PF_C2_QF} follows the similar steps as {\textbf{Algorithm~\ref{Alg_VQF_Case1}}}. While the convergence and complexity analysis are the same as in non-coplanar case, we omit these for brevity. 

\section{Numerical Results}
    In this section, we provide the numerical results to validate the analysis of radio sensing performance and demonstrate the advantages of the proposed VQF algorithms. A 2D UCA-based NF ISAC framework is considered. In particular, the BS is equipped with ${N}_{t}$ uniformly distributed transmit antenna operating at ${f}_{c}=28$ GHz without specific illustration. The aperture size of UCA is ${D}_{1}=\frac{ {N}_{t} d }{2 \pi}$. The noise power is set as ${\sigma}^{2}=-113$ dBm. The numerical results of the benchmark scheme are obtained as~\cite{FanLiu_22TSP,Zhaolin23_NF} under the UCA.
    \begin{figure}[!t]
    	\centering
    	\includegraphics[width=3.25in]{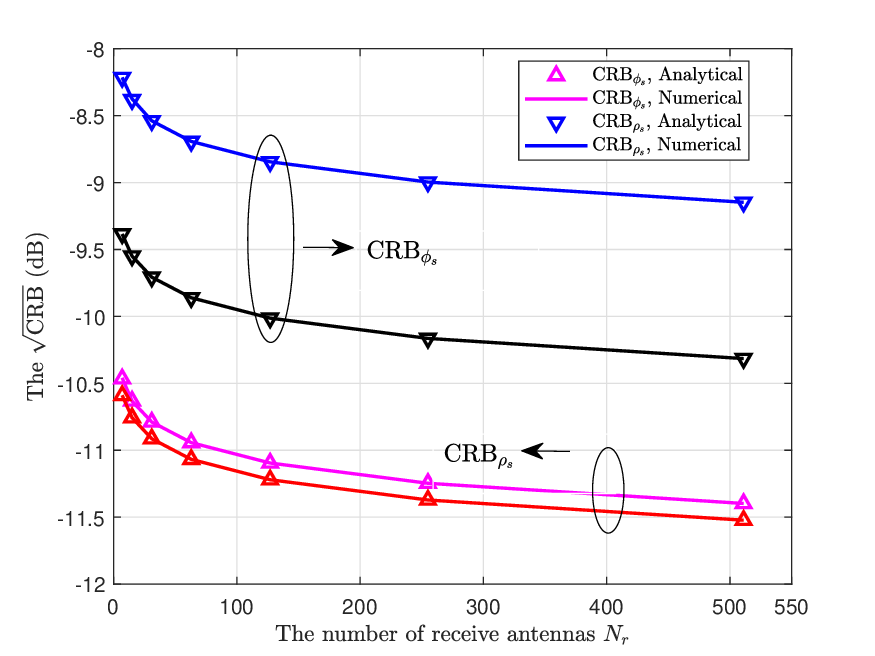}\\
    	\caption{The ${\rm{CRB}}$ performance versus different receive antennas under coplanar case.}\label{Fig:CRB_UCA}
    \end{figure}

    \vspace{-0.5cm}
    \subsection{Analysis Validation}
    \subsubsection{Radio sensing performance versus different number of receive antennas}
     In Fig.~\ref{Fig:CRB_UCA}, we study the radio sensing performance versus the number of receive antennas ${N}_{r}$ with ${\rho}_{s} < {d}_{{\rm{CA}}} $ and $ {\rho}_{s} \ge {R}_{{\rm{CA}}} $ under coplanar case. We set ${P}_{\max}=25$ dBm, ${N}_{t}=256$, and ${\bar{\gamma}}=1$ dB. The azimuth angle of ST is ${\phi}_{s}=30^{\circ} $ while the perpendicular distances are set as ${\rho}_{s}=15,40$ cm for two scenarios. With the increased number of receive antennas, the ${\rm{CRB}}_{{\phi}_{s}}$ reaches a lower value due to the large receive antennas size increasing the resolution of radio sensing. Among the two scenarios (${\rho}_{s}<{d}_{{\rm{CC}}} $ and ${\rho}_{s} \ge {R}_{{\rm{CA}}}$), the derived approximation expressions of ${\rm{CRB}}_{{\phi}_{s}} $ and $ {\rm{CRB}}_{{\rho}_{s}}$ are nearly equal to the numerical results which validate the efficiency of the approximate diagonal FIM. This phenomenon also indicates that the estimation of the angle and distance is nearly independent under the NF propagation with the symmetry circular antenna array.     

    \begin{figure}[!t]
    	\centering
    	\includegraphics[width=3.25in]{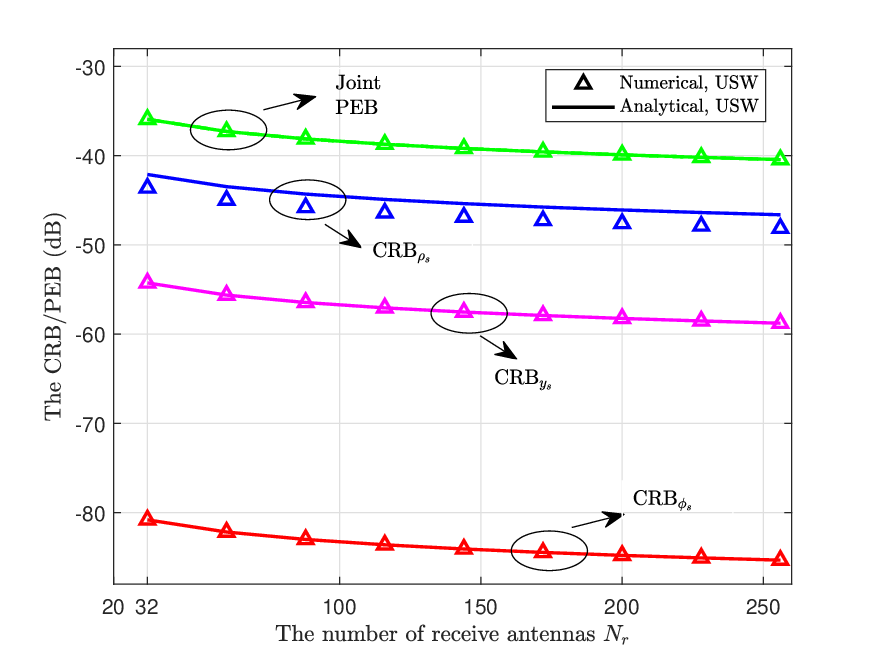} \\
    	\caption{The ${\rm{CRB}}$ performance versus different receive antennas ${N}_{r}$ under non-coplanar case.} \label{Fig:CRB_Case2}
    \end{figure}
    \subsubsection{Radio sensing performance versus different number of receive antennas under non-coplanar case}
    In Fig.~\ref{Fig:CRB_Case2}, we investigate the separate CRB w.r.t. ${ {\rho}_{s}, {\phi}_{s}, {y}_{s} }$ and joint SPEB performance with ${\mathbf{R}}_{x}=\frac{{{{{P}}_{\max }}}}{{{N_t}}} {\mathbf{I}}_{{N}_{t}\times {N}_{t}} $ versus different number of receive antennas under non-coplanar case. We set ${P}_{\max}=25$ dBm and ${N}_{t}=256$. The ST is located at ${\mathbf{p}}_{s}=[8,30^{\circ},2]^{\rm{T}} $. With the increase in the number of receive antennas, the CRB w.r.t. ${{\rho}_{s}, {\phi}_{s}, {y}_{s} } $ reduces, which indicates that the CRB is inversely proportional to the number of receive antennas as shown in \textbf{Proposition 4}. Compared with the ${\rm{CRB}}_{{\rho}_{s}} $ and ${\rm{CRB}}_{y_{s}}$, the value of ${\rm{CRB}}_{{\phi}_{s}} $ is smaller. The reason lies in that the decomposed main eigen vector of ${\rm{CRB}}_{{\rho}_{s}} $ and ${\rm{CRB}}_{y_s} $ is parallel to $ {\bar{{\bm{\alpha}}}}_{t} $ while the ${\bar{{\mathbf{h}}}}_{s}$ is composed of ${\bar{{\bm{\alpha}}}}_{t}$ and its orthogonal vector ${\bar{{\bm{\alpha}}}}_{t} \odot {\mathbf{v}}_{t,2,2}$. With the isotropic BF design at the BS, the joint ${\rm{SPEB}}$ is mainly determined by $ {\rm{CRB}}_{{\rho}_{s}} $ while the ${\rm{CRB}}_{{\phi}_{s}} $ and ${\rm{CRB}}_{{y}_{s}} $ affect the value of ${\rm{SPEB}}$ via the coordinate transform matrix. 

\begin{figure}[!t]
    \centering
    \includegraphics[width=3.25in]{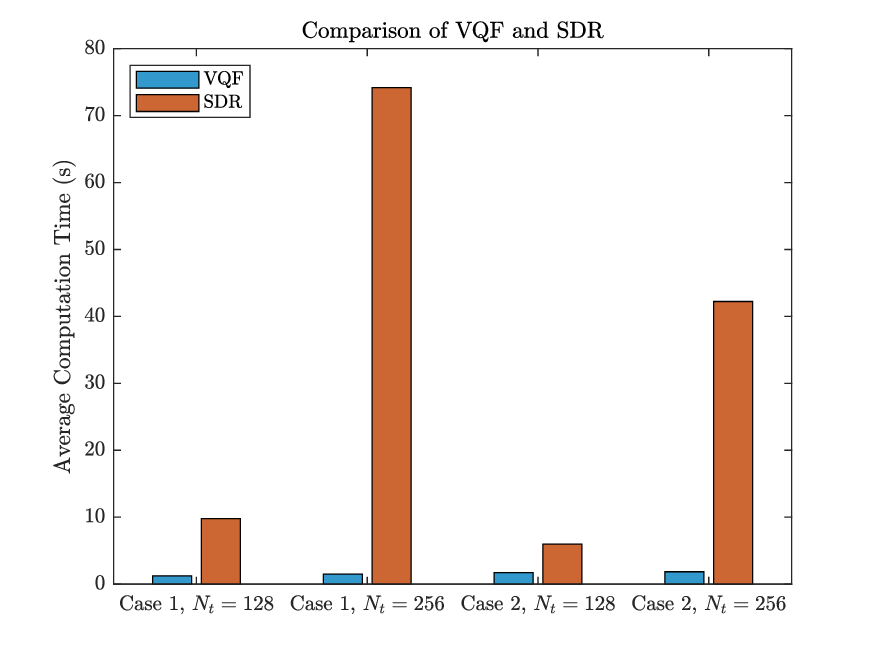}\\
    \caption{The Computation Complexity performance.}\label{Fig:RunningTime}
\end{figure}
\begin{figure}[!t]
    \centering
    \includegraphics[width=3.25in]{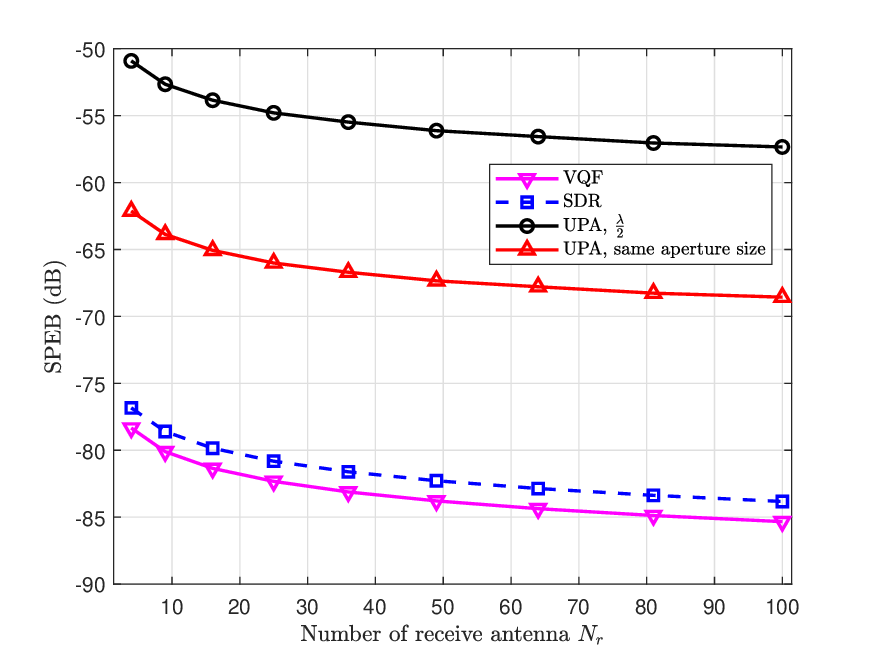}\\
    \caption{The SPEB performance versus ${N}_{r}$ under coplanar case.}\label{Fig:Optimality_Case1}
\end{figure}
\subsection{Algorithm Performance}
To evaluate the effectiveness of the proposed algorithms, three benchmarks are considered: 1) \textbf{SDR}: In this scheme, the BS is equipped with ${N}_{t}$ transmit antennas and ${N}_{r}$ receive antennas. The antenna interval is set to $\frac{{{\lambda_c}}}{2}$. The problem (5) is optimized via the SDR approach as \cite{FanLiu_22TSP}. 2) \textbf{UPA, $\frac{{{\lambda _c}}}{2}$}: In this scheme, the BS is equipped with a UPA transmit antennas. The antenna interval among each antennas is set to $\frac{{{\lambda_c}}}{2}$. The BF design is optimized by the SDR approach. 3) \textbf{UPA, same aperture size}: In this scheme, the UPA at the BS has the same aperture size as the UCA.   

\subsubsection{The Computation Complexity Performance}
In Fig.~\ref{Fig:RunningTime}, we compare the computation complexity of the proposed algorithm with different transmit antennas ${N}_{t}$ in both coplanar case and non-coplanar case. For coplanar case, the ST is in the coplane of UCA with ${\mathbf{p}}_{s}={[10,30^{\circ},0]^{\rm{T}}}$ while the ST under non-coplanar case is located at ${\mathbf{p}}_{s}={[8,30^{\circ},2]^{\rm{T}}}$. The receive antennas are set to ${N}_{r}=128$. The minimum communication requirement is set as ${\bar{\gamma}}=5$ dB. The initialization of ${\tilde{\mathbf{y}}}^{(0)}$ is determined by substituting ${\mathbf{w}}^{(0)}$ into \eqref{y_Case1_Update}. As shown in Fig.~\ref{Fig:RunningTime}, the proposed VQF consumes much less computation time than the \textbf{SDR} algorithm over four groups. The reason lies in that the proposed VQF solved the problem \eqref{Eq:PF_C1_QF} in a simplified vector form. The complexity of VQF is about ${\mathcal{O}}({N}_{t}^{2})$ while the \textbf{SDR} algorithm has the complexity at about ${\mathcal{O}}({N}_{t}^{3.5} )$. With the increase of the transmit antennas, the gap in the computation time between the two algorithms enlarges.      

\subsubsection{The SPEB Performance versus different ${N}_{r}$}
\begin{figure}[!t]
    \centering
    \includegraphics[width=3.25in]{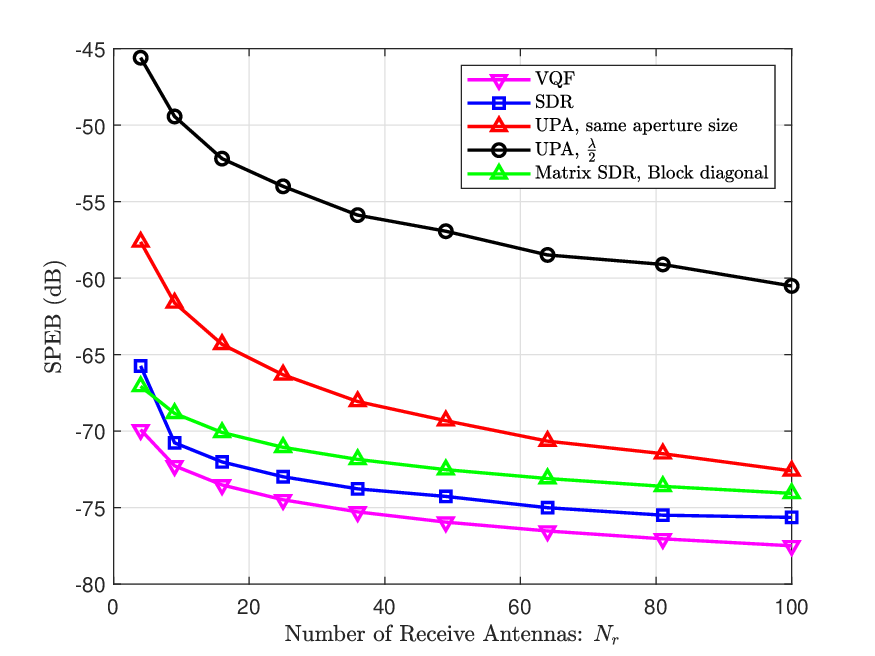}\\
    \caption{The SPEB performance versus ${N}_{r}$ under non-coplanar case.}\label{Fig:Optimality_Case2}
\end{figure}
In Fig.~\ref{Fig:Optimality_Case1}, we compare the SPEB performance of the proposed algorithm versus different receive antennas in coplanar case. The minimum communication requirement is set as ${\bar{\gamma}}=5$ dB and the number of transmit antennas is set as ${N}_{t}=256$. The aperture size is ${R}_{{\rm{CA}}}=0.4$ m. From Fig.~\ref{Fig:Optimality_Case1}, it can be seen that the SPEB decreased with the increment of the receive antennas. This is because that the SPEB is inversely proportional to ${N}_{r}$ as shown in {\textbf{Proposition~\ref{Pro:CRB_USW_UCA}}} and {\textbf{Proposition~\ref{Pro_FIM_Case2}}}. Compared with the UPA schemes, the proposed VQF algorithm achieves at least a 15 dB gain under coplanar case. This phenomenon illustrates that the UCA is more suitable for the coplanar ST. Meanwhile, the proposed VQF algorithm achieves a 2 dB gain than the \textbf{SDR} algorithm.

In Fig.~\ref{Fig:Optimality_Case2}, we compare the SPEB performance versus different receive antennas in non-coplanar case. The \textbf{SDR, Block diagonal} refers to the SDR algorithm to solve the problem \eqref{Eq:PF_SDR_Case2} with the initial FIM. The transmit power is set to ${P}_{\max}=35$ dBm and the communication requirement is set to ${\bar{\gamma}}=10$ dB. The number of transmit antennas is ${N}_{t}=256$ with the aperture size $R_{\rm{CA}} \approx 0.3654$ m. The ST is located at ${\mathbf{p}}_{s}=[4,30^{\circ},1 ]^{{\rm{T}}}.$ With the increase of the receive antennas, the SPEB performance elevated over all five schemes. Different from coplanar case, the performance gain between the UCA and \textbf{UPA, same aperture size} narrowed with more receive antennas. Compared with the \textbf{SDR, Block diagonal}, the proposed extended VQF algorithm achieves better SPEB performance. This is because the approximate diagonal FIM discard the terms $[{\mathbf{J}}_{{\bm{\gamma}}} ]_{(2,3)} $ and $[{\mathbf{J}}_{{\bm{\gamma}}} ]_{(3,2)}  $, which denotes the correlation between the perpendicular distance and the signed distance. As it has shown in {\textbf{Lemma \ref{Pro:Case2_decompose}}}, the optimal solution of problem \eqref{Eq:PF_C2_QF} satisfies ${\mathbf{w}}^{*} \in {\rm{span}} ( {\bm{\alpha}}_{t}, {\bm{\alpha}}_{t} \odot {\bar{\mathbf{h}}}_{s} )$. With the block diagonal FIM, the weight to minimize the ${\rm{CRB}}_{{\phi}_{s}} $ in the objective function decreased. While the ${\rm{CRB}}_{{\phi}_{s}}$ is much larger than ${\rm{CRB}}_{{\rho}_{s}} $ and ${\rm{CRB}}_{{y}_{s}} $, which indicates that the ${\rm{CRB}}_{{\phi}_{s}}$ plays a dominant role to determine the joint SPEB performance. 

\subsubsection{The SPEB Performance versus different $P_{\max}$}
\begin{figure}[!t]
    \centering
    \includegraphics[width=3.25in]{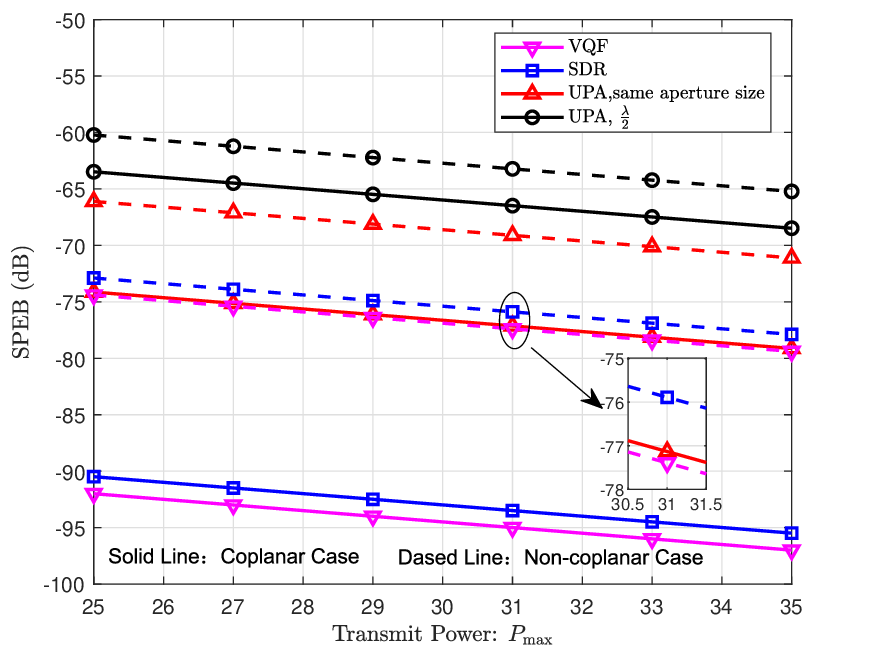}\\
    \caption{The SPEB performance versus different ${P}_{\max}$.}\label{Fig:SPEB_Pmax}
\end{figure}
\begin{figure}[!t]
    \centering
    \includegraphics[width=3.25in]{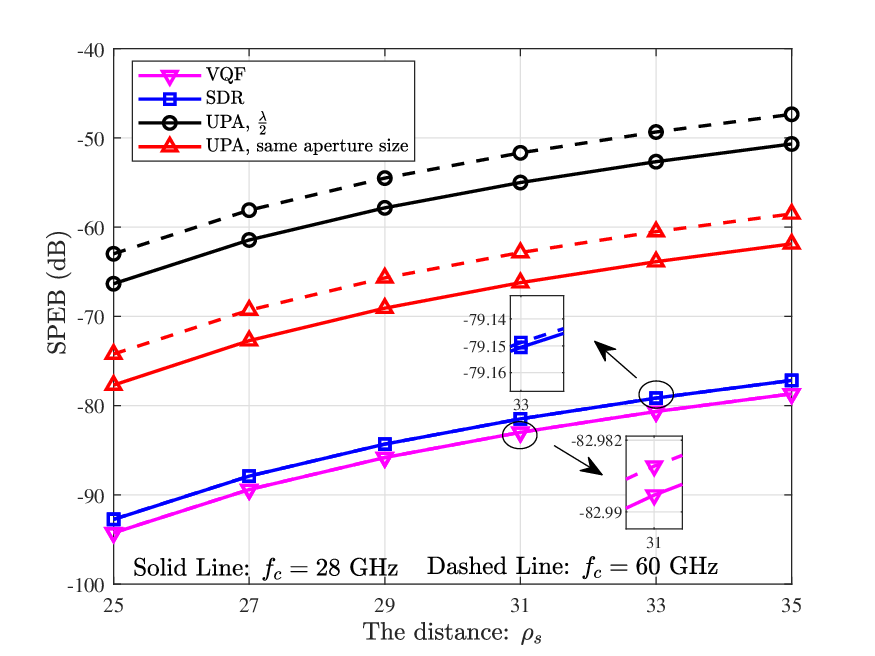}\\
    \caption{The SPEB performance versus different ${\rho}_{s}$ under coplanar case.}\label{Fig:SPEB_rhos}
\end{figure}
In Fig.~\ref{Fig:SPEB_Pmax}, we compare the SPEB performance of the proposed algorithms versus different transmit power ${P}_{\max}$. The ST is located at ${\mathbf{p}}_{s}= {[10,30^{\circ},0]^{{\rm{T}}}}$ in coplanar case while the ST in non-coplanar case is located at ${\mathbf{p}}_{s}= [8,30^{\circ},2]^{{\rm{T}}}$. The minimum communication requirement is set as ${\bar{\gamma}}=5$ dB. The number of transmit and receive antennas are set to ${N}_{t}=256$ and ${N}_{t}=100$, respectively. In Fig.~\ref{Fig:SPEB_Pmax}, it can be seen that the transmit power ${P}_{\max}$ is inversely proportional to the SPEB in dB under coplanar case and non-coplanar case, which is consistent with the \textbf{Proposition~\ref{Pro:CRB_USW_UCA}}. Compared with the \textbf{UPA, $\frac{{{\lambda_c}}}{2}$}, the UCA schemes reach at least 20 dB gain in SPEB. This is because the UCA has a larger aperture size than UPA with the same antenna interval. When the UPA scheme has the same aperture size as UCA, there is still about 10 dB gain in SPEB as the effective aperture size of UCA is larger than UPA. Therefore, the antenna array with a larger effective aperture size is prone to achieve a better radio sensing performance. Meanwhile, the proposed VQF in coplanar case has less position error than in non-coplanar case, which indicates that the UCA antenna array is more suitable for the ST in the coplanar case. 

\begin{figure}[!t]
    \centering
    \vspace{-0.2cm}
    \includegraphics[width=3.25in]{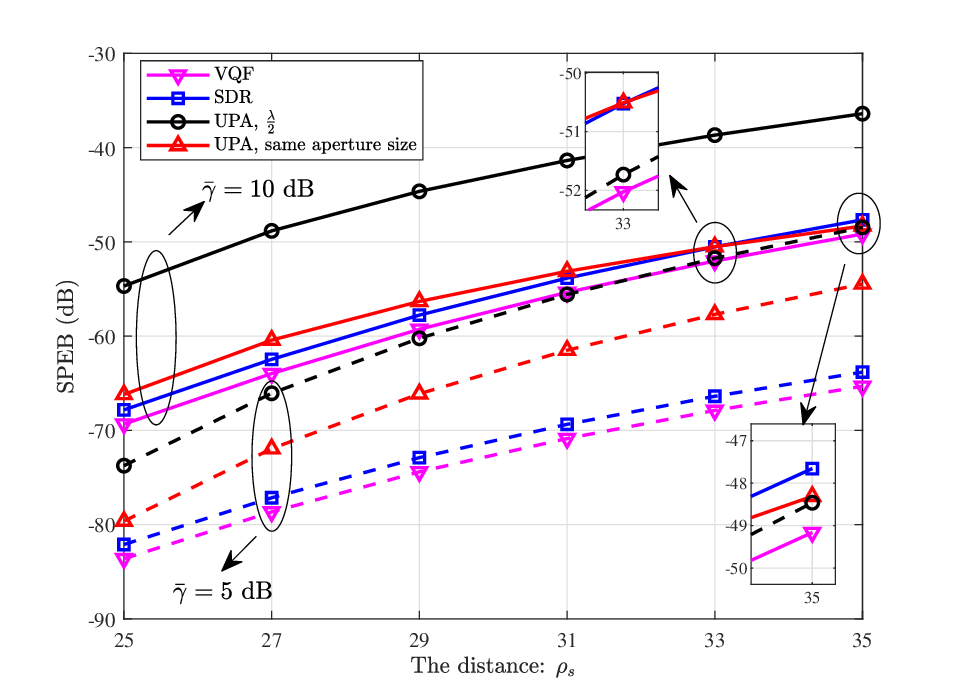}\\
    \caption{The SPEB performance versus different ${\rho}_{s}$ under non-coplanar case.} \label{Fig:SPEB_rhos_Case2}
\end{figure}
\subsubsection{The SPEB Performance versus different ${\rho}_{s} $ } 

In Fig.~\ref{Fig:SPEB_rhos}, we compare the SPEB performance of the proposed algorithm versus different ${\rho}_{s} $ in coplanar case, where ${f}_{c}=28$ GHz and $ {f}_{c}=60$ GHz are considered. The ST moves along ${\mathbf{p}}_{s}= {[{\rho}_{s},30^{\circ},0]^{\rm{T}}}$ with different ${\rho}_{s}$. The minimum communication requirement is set as ${\bar{\gamma}}=5$ dB and the transmit power is set as ${P}_{\max}=25$ dBm. The number of transmit and receive antennas are set to ${N}_{t}=256$ and ${N}_{t}=100$, respectively. The aperture size of UCA is ${R}_{{\rm{CA}}}=0.4365$ m. With the increased ${\rho}_{s}$, the SPEB performance got worse overall schemes. This is because the extended propagation distance leads to a decrease in both the strength of the reflected echo signal and the coefficient reflection ${\alpha}_{s}$. Compared with the other benchmarks, the proposed VQF achieves the best radio sensing performance due to the larger effective array aperture and the near-optimal solution-based eigenvalue decomposition. With the increased ${f}_{c}$, the SPEB increased because the effective propagation distance is reduced with the small wavelength ${\lambda}_{c}$, which indicates there is a tradeoff between the carrier frequency and the aperture size. Compared with the UPA schemes, the UCA schemes have a small variation over different carrier frequency ${f}_{c}$, which indicates that the UCA is more flexible among different carrier frequency.  

In Fig.~\ref{Fig:SPEB_rhos_Case2}, we compare the SPEB performance of the proposed algorithm versus different ${\rho}_{s}$ in non-coplanar case. Two scenarios are considered: ${\bar{\gamma}}=10$ dB and ${\bar{\gamma}}=5$dB. The ST moves along ${\mathbf{p}}_{s}= {[{\rho}_{s},30,1]^{\rm{T}}}$ with different ${\rho}_{s}$. The transmit power budget is set as ${P}_{\max}=25$ dBm. The number of transmit and receive antennas are set to ${N}_{t}=256$ and ${N}_{r}=100$, respectively. Similar to Fig.~\ref{Fig:SPEB_rhos}, the SPEB performance increase with the increased ${\rho}_{s}$. The proposed VQF achieves the best SPEB performance over the other benchmarks over two scenarios. When ${\bar{\gamma}}=5$dB, the SPEB performance gain increased with the increment in ${\rho}_{s}$, which indicates that the UCA is more suitable for estimation of far ST with the mild SINR scenario. With the increased communication requirement, the gap between the \textbf{UPA, $\frac{{\lambda}_{c}}{2}$} and the other benchmarks is enlarged, which indicates that the aperture size of the antenna array is more important than the antenna array geometry with the stringent communication requirement. 

\section{Conclusion} 
In this paper, a UCA-based NF ISAC framework has been proposed, where the USW model was invoked to unveil the dependence of antenna array design on radio sensing performance. The separate CRB over 3D coordinates and joint SPEB approximation were derived when the ST is located in the coplanar or the non-coplanar of the UCA. Constrained by the communication requirement and power budget, a SPEB minimization problem was formulated to optimize the BF design. When the ST is in the coplanar of the UCA, a low complexity VQF algorithm was proposed to minimize the joint SPEB by exploiting the dominant component of the partial derivative. When the ST is located in the non coplanar of the UCA, a block diagonal-based SDR method and an extended low-complexity VQF algorithm were proposed. Numerical results demonstrated that the diagonal FIM can be utilized to evaluate the joint SPEB performance in NF-ISAC systems with the negligible mismatch. The proposed VQF algorithms outperform the conventional SDR algorithm in the joint radio sensing performance and complexity. When ST located in the coplanar case and the perpendicular distance is more than twice of the signed distance in non-coplanar case, the UCA achieves better performance than UPA benefiting from the effective Rayleigh distance over different azimuth angle. 

\appendix

\subsection{Proof of {\textbf{Lemma}} \ref{Lemma_Case1_norm}}
\renewcommand{\theequation}{A.\arabic{equation}}
\setcounter{equation}{0} \label{App:Lemma_Case1_norm} 

Based on the definition of ${\mathbf{v}}_{l,2} $, it is clear that 

\begin{small}
	\vspace{-0.4cm}
	\begin{subequations}
		\begin{equation}
			\begin{split}
				 {\left\| { { {\mathbf{v}}_{l,2}} } \right\|_1}
				& \mathop \approx \limits^{(a)}  \frac{{N}_{l} }{2\pi} \int\limits_{0}^{2\pi } {\frac{{ {\rho}_{s}  - \cos x}}{{ \sqrt { {R}_{{\rm{CA}}}^{2} - 2 {\rho}_{s} {R}_{{\rm{CA}}} \cos x + {{\rho}_{s}^2}} }}} dx - {N_l} \\
				&  = {N_l} {\bm{\Upsilon}} (\frac{{{\rho_s}}}{{{R_{{\rm{CA}}}}}})  - {N_l},
			\end{split} 
		\end{equation}
		\begin{equation}\label{app_v_2}
			\begin{split}
				{\left\| { { {\mathbf{v}}_{l,2}} } \right\|_{2}^{2}}
				& = \sum\limits_{m=- \frac{{N}_{l}-1}{2} }^{ \frac{{N}_{l}-1}{2}} {\frac{{{{({\rho _s} - {R_{{\rm{CA}}}}\cos (\frac{{2\pi m }}{N} ))}^2}}} {{(R_{{\rm{CA}}}^2 + \rho _s^2 - 2{R_{{\rm{CA}}}}{\rho _s}\cos (\frac{{2\pi m }}{N} ))}}} \\
				& \; + {N_r} - 2{N_r} {\bm{\Upsilon}} (\frac{{{\rho _s}}}{{{R_{{\rm{CA}}}}}}),
			\end{split}
		\end{equation}
	\end{subequations}
\end{small}

\noindent where $\mathop \approx \limits^{(a)} $ comes from $\frac{{- a\cos (\frac{{2\pi m }}{N_{l}} )}}{{{{(a^2 + b^2 - 2ab\cos (\frac{{2\pi m }}{N_{l}}))}^{\frac{1}{2}}}}} \approx  \int\limits_{0}^{2\pi } {\frac{{ - R_{{\rm{CA}}} \cos x}}{{ \sqrt { {R}_{{\rm{CA}}}^{2} - 2 {\rho}_{s} {R}_{{\rm{CA}}} \cos x + {{\rho}_{s}^2}} }}} dx $ when $ {\delta} \gg 1$ and ${N}_{r} \gg 2\pi $. Similarly, the expression of $ {\Delta}_{1}  \buildrel \Delta \over =  \sum\limits_{m=- \frac{{N}_{l}-1}{2} }^{ \frac{{N}_{l}-1}{2}} { \frac{{{{ {{R}_{{\rm{CA}}}}^{2} \cos^{2}(\frac{{2\pi m }}{N} )  }}}} {{ \left[  {\rho}_{{\rm{CA}}}^{\left(m\right)} \right]^{2}  }} } $ is given by 

\begin{small}
	\begin{equation}
            {\Delta}_{1} \approx  \left\{ \begin{array}{l}
				\frac{{R_{{\rm{CA}}}^2{N_{l}}}}{{(R_{{\rm{CA}}}^2 - \rho _s^2)}} - \frac{{{N_{l}}}}{2},  \; {\rho _s} < {R_{{\rm{CA}}}}\\
				\frac{{R_{{\rm{CA}}}^2{N_r}(\rho _s^2 + R_{{\rm{CA}}}^2)}}{{2\rho _s^2(\rho _s^2 - R_{{\rm{CA}}}^2)}}, \; {\rho _s} > {R_{{\rm{CA}}}}
			\end{array} \right. . 
	\end{equation}
\end{small}

\noindent By substituting the mentioned equations into \eqref{app_v_2}, \eqref{v_2_int_2} can be derived.

\subsection{Proof of {\textbf{Lemma}} \ref{Lemma_Case2_norm} }
\renewcommand{\theequation}{B.\arabic{equation}}
\setcounter{equation}{0}
\label{App:Lemma_Case2_norm} 
 We first take $\left\| {{\mathbf{v}}_{l,2,2}} \right\|_1$ as an example to illustrate \eqref{Auxiliary_Periodical_0}. Due to the periodical property of sine function, $\left\| {{\mathbf{v}}_{l,2,2}} \right\|_1 $ can be derived as
        \begin{equation}\label{Case2_norm1_2}
           \begin{split}
               & \left\| {{\mathbf{v}}_{l,2,2}} \right\|_1 = \sum\limits_{m =  - \frac{{{N_l} - 1}}{2}}^{m = \frac{{{N_l} - 1}}{2}} {\frac{{\sin ({\phi _s} - \phi _s^{(m)})}}{{\sqrt {{{(\rho_{{\rm{CA}}}^{(m)})}^2} + y_s^2} }}} \\
               & \approx \frac{{{N_{l}}}}{{2\pi }}\int\limits_0^{2\pi } {\frac{{\sin x}} {{\sqrt {R_{{\rm{CA}}}^2 + \rho_s^2 + y_s^2 - 2{R_{{\rm{CA}}}}{\rho_s}\cos x} }}dx} = 0.
           \end{split}
        \end{equation}It is noted that $ \left\| {{{\mathbf{v}}_{l,2,1}} \odot {{\mathbf{v}}_{l,2,2}}} \right\|_2^1$ and $ \left\| {{{\mathbf{v}}_{l,2,3}} \odot {{\mathbf{v}}_{l,2,2}}} \right\|_2^1$ both contain a similar structure to \eqref{Case2_norm1_2}. The \eqref{Auxiliary_Periodical_0} can be derived. Based on the definition of ${{\mathbf{v}}_{l,2,3}} $, $\left\| {{\mathbf{v}}_{l,2,3}} \right\|_1$ can be further derived as $ \left\| {{\mathbf{v}}_{l,2,3}} \right\|_1=\sum\limits_{m =  - \frac{{{N_{l}} - 1}}{2}}^{m = \frac{{{N_{l}} - 1}}{2}} {\frac{1}{{\sqrt {\rho _{\rm{s}}^2 + y_s^2} }} - \frac{1}{{\sqrt {{{(\rho _{{\rm{CA}}}^{(m)})}^2} + y_s^2} }}}=\frac{{{N_{l}}}}{{\sqrt {\rho _{s}^2 + y_s^2} }} - \sum\limits_{m =  - \frac{{{N_{l}} - 1}}{2}}^{m = \frac{{{N_{l}} - 1}}{2}} {\frac{1}{{\sqrt {{{(\rho _{{\rm{CA}}}^{(m)})}^2} + y_s^2} }}}$. Let ${\gamma}_{1}=2{R}_{{\rm{CA}}} {\rho}_{s}, {\gamma}_{2}={y}_{s}^{2}+R_{{\rm{CA}}}^{2}+{\rho}_{s}^{2}$, we have $\left\| {{\mathbf{v}}_{l,2,3}} \right\|_1 \mathop  = \limits^{(a)} \frac{{{N_l}}}{{\sqrt {\rho _{\rm{s}}^2 + y_s^2} }} - \frac{{2{N_l}}}{{\pi \sqrt {{\gamma _2} + {\gamma _1}} }}{\mathbf{K}}\left( {\sqrt {\frac{{2{\gamma _1}}}{{{\gamma _2} + {\gamma _1}}}} } \right)$, where $\mathop  = \limits^{(a)}$ comes from (Eq.3.671) in~\cite{Integral_reference} and ${\mathbf{K}} \left( \alpha \right)=\int_0^{\pi /2} {\frac{1}{{\sqrt {1 - {{\alpha}^2}{{\sin }^2}\theta } }}} {\kern 1pt} d\theta $ denotes the complete elliptic integral of the first kind~\cite{Integral_reference}. Similarly, $\left\| {{{\mathbf{v}}_{l,2,2}}} \right\|_2^2 $ can be derived as $ \left\| {{{\mathbf{v}}_{l,2,2}}} \right\|_2^2 = \sum\limits_{m =  - \frac{{{N_{l}} - 1}}{2}}^{m = \frac{{{N_{l}} - 1}}{2}} {\frac{{{{\sin }^2}({\phi _s} - \phi _s^{(m)})}}{{{{(\rho_{{\rm{CA}}}^{(m)})}^2} + y_s^2}}}  \approx \frac{{{N_l}}}{{2\pi }}\int\limits_0^{2\pi } {\frac{{{{\sin }^2}x}}{{{\gamma _2} - {\gamma _1}\cos x}}} dx = \frac{{{N_l}}}{{4\pi }}\int\limits_0^{2\pi } {\frac{{(1 - \cos 2x)}}{{{\gamma _2} - {\gamma _1}\cos x}}} dx.$. Recalling that $ \int\limits_0^\pi  {\frac{1}{{1 \pm a\cos x}}dx}  = \frac{\pi }{{\sqrt {1 - {a^2}} }}$ and $\int\limits_0^\pi  {\frac{{\cos 2x}}{{1 \pm a\cos x}}dx}  = \frac{\pi }{{\sqrt {1 - {a^2}} }}\frac{{{{[\sqrt {1 - {a^2}}  - 1]}^2}}}{{{a^2}}} $~\cite{Integral_reference}, we have $ \left\| {{{\mathbf{v}}_{l,2,2}}} \right\|_2^2 = \frac{{{N_{l}}}}{{\gamma _1^2}} {\rm{( {\gamma _2} - \sqrt {\gamma _2^2 - \gamma _1^2} )}}.$ Following the similar derivation, \eqref{Case2_norm2_1} and \eqref{Case2_norm2_3} can also be derived. 

\subsection{Proof of {\textbf{Proposition~\ref{Pro_FIM_Case2}}} }
\label{App:Pro_FIM_Case2}

\renewcommand{\theequation}{C.\arabic{equation}}
\setcounter{equation}{0}
Considering that ${\mathbf{A}}_{2}={\bm{\alpha}}_{t,2}^{{\rm{H}}}{\bm{\alpha}}_{r,2}$ and the derivatives ${\dot{{\bm{\alpha}}}}_{l,i}$ are given by \eqref{Case2_deriva_0}, the partial derivative of ${\mathbf{A}}_{2}$ over $\left\{{ {\phi}_{s}, {\rho}_{s}, {y}_{s}} \right\} $ can be expressed as
            \begin{subequations}
                \begin{equation}
                    {\dot{{\mathbf{A}}}}_{2,{\rho _s}} = \frac{{\partial {{\mathbf{A}_2}}}} {{\partial {\rho _s}}}= \frac{{j2\pi }}{{{\lambda _c}}}[{\rm{diag}}({{\mathbf{v}}_{r,2,1}}){{\mathbf{A}}_2} - {{\mathbf{A}}_2} {\rm{diag}}({{\mathbf{v}}_{t,2,1}}) ],
                \end{equation}
                \begin{equation}
                    {\dot{{\mathbf{A}}}}_{2,{\phi_s}} = \frac{{\partial {{\mathbf{A}_2}}}}{{\partial {\phi _s}}}=\frac{{j\pi {\gamma}_{1} }} {{{\lambda_c}}}[{{\mathbf{A}}_2} {\rm{diag}}({{\mathbf{v}}_{t,2,2}}) - {\rm{diag}}({{\mathbf{v}}_{r,2,2}}){{\mathbf{A}}_2}],
                \end{equation}
                \begin{equation}
                    {\dot{{\mathbf{A}}}}_{2,{y_s}}=\frac{{\partial {{\mathbf{A}_2}}}}{{\partial {y_s}}}= \frac{{j2\pi {y_s}}}{{{\lambda _c}}}[ {\rm{diag}}({{\mathbf{v}}_{r,2,3}}){ {\mathbf{A}}_2} - {{\mathbf{A}}_2} {\rm{diag}}({{\mathbf{v}}_{r,2,3}}) ].
                \end{equation}
            \end{subequations}Recall that \eqref{FIM_element0}, the FIM elements $\left[ {{\mathbf{J}}_{{\bm{\gamma}}}} \right]_{i,j}$ can be expressed as $ \left[ {{\mathbf{J}}_{{\bm{\gamma}}}} \right]_{i,j}= \frac{{2{{\left| {{\alpha_{{s}}}} \right|}^2} L {\Re [{\rm{Tr}} ( { {\mathbf{R}}_{x} }  {\dot{\mathbf{A}}}_{2,i}^{\rm{H}} {{\dot{\mathbf{A}}}_{2,j}} )] } }} {{{\sigma ^2}}}$. We first illustrate the property of the cross terms where $i \neq j$. Considering that $\left[ {{\mathbf{J}}_{{\bm{\gamma}}}} \right]_{1,2} $ and $ \left[ {{\mathbf{J}}_{{\bm{\gamma}}}} \right]_{2,3} $ depends on the value of ${\dot{\mathbf{A}}}_{{\rho_s}}^{\rm{H}} {{\dot{\mathbf{A}}}_{{\phi_s}}} $ and $ {\dot{\mathbf{A}}}_{{\phi_s}}^{\rm{H}} {{\dot{\mathbf{A}}}_{{y_s}}}$, respectively,  ${\dot{\mathbf{A}}}_{{\rho_s}}^{\rm{H}} {{\dot{\mathbf{A}}}_{{\phi_s}}} $ is a linear function of $\left\| {{{\mathbf{v}}_{l,2,1}} \odot {{\mathbf{v}}_{l,2,2}}} \right\|_2^1$ and ${\dot{\mathbf{A}}}_{{\phi_s}}^{\rm{H}} {{\dot{\mathbf{A}}}_{{y_s}}}$ is a linear function of $\left\| {{{\mathbf{v}}_{l,2,1}} \odot {{\mathbf{v}}_{l,2,3}}} \right\|_2^1 $. Based on \eqref{Auxiliary_Periodical_0}. we have $\left[ {{\mathbf{J}}_{{\bm{\gamma}}}} \right]_{1,2}=\left[ {{\mathbf{J}}_{{\bm{\gamma}}}} \right]_{2,1}^{\rm{H}} =\left[ {{\mathbf{J}}_{{\bm{\gamma}}}} \right]_{2,3}=\left[ {{\mathbf{J}}_{{\bm{\gamma}}}} \right]_{3,2}^{\rm{H}} \approx 0$. Thus, the FIM can be approximated via the block diagonal matrix as shown in \eqref{FIM_bd}. According to the Matrix Inverse Lemma~\cite{MatrixTheory}, the inverse of $ {\mathbf{J}}_{ {\mathbf{p}}_{s} } $ can be derived as 
        \begin{equation}
        	{\mathbf{J}}_{ {\mathbf{p}}_{s} }^{-1}=\frac{ {D}_{ {\mathbf{p}}_{s} } }{ \left[ {\mathbf{J}}_{ {\mathbf{p}}_{s} } \right]_{{\rm{11}}} \left[ {\mathbf{J}}_{ {\mathbf{p}}_{s} } \right]_{{\rm{22}}} \left[ {\mathbf{J}}_{ {\mathbf{p}}_{s} } \right]_{{\rm{33}}} },
        \end{equation}where ${D}_{ {\mathbf{p}}_{s} } $ denotes the cofactor matrix of ${\mathbf{J}}_{ {\mathbf{p}}_{s} }$. Based on \textbf{Lemma 1}, we have that $ {\mathbf{A}}_{2}^{\rm{H}} {\rm{diag}} \left( {\mathbf{v}}_{r,1} \right) {\mathbf{A}}_{2} = {\mathbf{A}}_{2}^{\rm{H}} {\rm{diag}} \left( {\mathbf{v}}_{r,2} \right) {\mathbf{A}}_{2}  = 0 $ and ${\mathbf{A}}_{2}^{\rm{H}} {\rm{diag}}({\mathbf{v}}_{r,2,3}){{\mathbf{A}}_{2} }= \frac{ {{R_{{\rm{CA}}}}{N_{r}}} {{\bar{{\mathbf{A}}}}_{2}} }{2} $. Let ${\bar{\mathbf{A}}}_{2}={\bm{\alpha}}_{t,2} {\bm{\alpha}}_{t,2}^{\rm{H}} $, $\left\{ {\dot{\mathbf{A}}}_{{{\phi}_s}}^{\rm{H}} {{\dot{\mathbf{A}}}_{{\phi_s}}}, {\dot{\mathbf{A}}}_{{\rho_s}}^{\rm{H}} {{\dot{\mathbf{A}}}_{{\rho_s}}}, {\dot{\mathbf{A}}}_{{y_s}}^{\rm{H}} {{\dot{\mathbf{A}}}_{{y_s}}}, {\dot{\mathbf{A}}}_{{\rho_s}}^{\rm{H}} {{\dot{\mathbf{A}}}_{{y_s}}}  \right\} $ can be expressed as \eqref{General_AA_Diagonal}. With the isotropic BF design at the BS, ${\rm{Tr}}({{\mathbf{R}}_{x}} {\dot{{\mathbf{A}}}}_{{\rho_s}}^{\rm{H}} {{\dot{{\mathbf{A}}}}_{{\rho_s}}}) $ and ${\rm{Tr}}({{\mathbf{R}}_{x}} {\dot{{\mathbf{A}}}}_{{y_s}}^{\rm{H}} {{\dot{{\mathbf{A}}}}_{{y_s}}}) $ can be further derived as \eqref{Trace_Case2_USW}, where ${\phi}_{{\rho}_{s}} {\rm{(}}{\phi _s}{\rm{,}}{\rho _s}{\rm{,}}{{\rm{y}}_s}{\rm{)}} =\frac{{(\rho _s^2 + \frac{{R_{{\rm{CA}}}^2}}{2} - {\gamma _2} + \frac{{\gamma _2^2}}{{8\rho _s^2}})}}{{\sqrt {\gamma _2^2 - \gamma _1^2} }} + \frac{{\sqrt {\gamma _2^2 - \gamma _1^2} }}{{8\rho _s^2}} - \frac{{4\rho _s^2{K^2}\left( {\sqrt {\frac{{2{\gamma _1}}}{{{\gamma _2} + {\gamma _1}}}} } \right)}}{{{\pi ^2}({\gamma _2} + {\gamma _1})}} + \frac{{{\gamma _1}\kappa {\rm{(}}{\gamma _2}{\rm{,}}{\gamma _1}{\rm{)}}K\left( {\sqrt {\frac{{2{\gamma _1}}}{{{\gamma _2} + {\gamma _1}}}} } \right)}}{{{\pi ^2}\sqrt {{\gamma _2} + {\gamma _1}} }} - \frac{{R_{{\rm{CA}}}^2{\kappa ^2}{\rm{(}}{\gamma _2}{\rm{,}}{\gamma _1}{\rm{)}}}}{{4{\pi ^2}}} $. $\mathop  = \limits^{(a)} $ comes from that $ {\rm{Tr}}( {\rm{diag}}({{\mathbf{v}}_{t,2,i}}){{\bar{{\mathbf{A}}}}_{2}} {\rm{diag}}({{\mathbf{v}}_{t,2,i}})) = \left\| {{{\mathbf{v}}_{t,2,i}}} \right\|_2^2, i \in \left\{1,2,3\right\}$ and $\mathop  = \limits^{(b)} $ comes from that ${\rm{Tr}}({\rm{diag}}({\mathbf{v}}_{r,2,3}){{\bar{{\mathbf{A}}}}_{2}})={\rm{Tr}}({{\bar{{\mathbf{A}}}}_{2}} {\rm{diag}}({\mathbf{v}}_{r,2,3}))=\frac{{{R_{{\rm{CA}}}}{N_{t}}}}{2} $,  ${\rm{Tr}}({\rm{diag}}({\mathbf{v}}_{t,2,3}){{\bar{{\mathbf{A}}}}_{2}} {\rm{diag}}({\mathbf{v}}_{t,2,3}))=\frac{ {N_{t}} ({R_{{\rm{CA}}}^2}-2{\rho_s^2} )}{4}$. By substituting \eqref{Trace_Case2_USW} into \eqref{General_CRB_Ori}, \eqref{General_CRB_Iso_USW} can be derived. 
        
        \begin{figure*}
       \begin{small}
    \begin{subequations}\label{General_AA_Diagonal}
                \begin{equation}
        			{\dot{{\mathbf{A}}}}_{2,{\phi_s}}^{\rm{H}} {\dot{{\mathbf{A}}}}_{2,{\phi _s}} =\frac{{4{\pi ^2}R_{{\rm{CA}}}^2\rho_s^2}}{{\lambda_c^2}}[{N_{r}} {\rm{diag}}({{\mathbf{v}}_{t,2,2}}){ {\bar{{\mathbf{A}}}}_{2}} {\rm{diag}}({{\mathbf{v}}_{t,2,2}}) + \left\| {{{\mathbf{v}}_{r,2,2}}} \right\|_2^2{{\bar{{\mathbf{A}}}}}_{2} ],
        		\end{equation}
        		\begin{equation}
        			\begin{split}
        				{\dot{{\mathbf{A}}}}_{2,{\rho _s}}^{\rm{H}} {\dot{{\mathbf{A}}}}_{2,{\rho_s}} =& \frac{{4{\pi ^2}}}{{\lambda _c^2}}\{ {\left\| {{{\mathbf{v}}_{r,2,1}}} \right\|_2} {{\bar{{\mathbf{A}}}}_{2}} - {\left\| {{{\mathbf{v}}_{r,2,1}}} \right\|_1}[ {\rm{diag}}({{\mathbf{v}}_{t,2,1}}){{\bar{{\mathbf{A}}}}_{2}} + { {\bar{{\mathbf{A}}}}_{2}} {\rm{diag}}({{\mathbf{v}}_{t,2,1}}) ] + {N_r}{\rm{diag}}({{\mathbf{v}}_{t,2,1}}){{\bar{{\mathbf{A}}}}_{2}} {\rm{diag}}({{\mathbf{v}}_{t,2,1}}) \}, 
        			\end{split}
        		\end{equation}
        	   \begin{equation}
        	   	\begin{split}
        	   	  {\dot{{\mathbf{A}}}}_{2,{y_s}}^{\rm{H}} {\dot{{\mathbf{A}}}}_{2,{y_s}} & = \frac{{4{\pi ^2}y_s^2}}{{\lambda _c^2}}\{ \left\| {{{\mathbf{v}}_{r,2,3}}} \right\|_2^2{{\bar{{\mathbf{A}}}}_{2}} - {\left\| {{\mathbf{v}}_{r,2,3}} \right\|_1}[{\rm{diag}}({{\mathbf{v}}_{t,2,3}}) {{\bar{{\mathbf{A}}}}_{2}} + { {\bar{{\mathbf{A}}}}_{2}} {\rm{diag}}({{\mathbf{v}}_{t,2,3}})]+ {{N}_r} {\rm{diag}} ({{\mathbf{v}}_{t,2,3}}) {{\bar{{\mathbf{A}}}}_{2}} {\rm{diag}} ({{\mathbf{v}}_{t,2,3}})\}, 
        	   	\end{split}
        	   \end{equation}
                \hrulefill
        	\end{subequations}
           \end{small}	
        \end{figure*}
        
        \begin{figure*}   
         \begin{small}
             \begin{subequations}\label{Trace_Case2_USW}
                \begin{equation}
                    {\rm{Tr}} \left[ {  {{\mathbf{R}}_{x} } {\dot{{\mathbf
                    A}}}_{2,{\phi_s}}^{\rm{H}} {{\dot{{\mathbf
                    A}}}_{2,{\phi_s}}} } \right] \mathop  = \limits^{(a)}  \frac{{8{\pi ^2}R_{{\rm{CA}}}^2\rho _s^2{P_{\max }}{N_r}({\gamma _2} - \sqrt {\gamma _2^2 - \gamma _1^2} )}}{{\lambda _c^2\gamma _1^2}}, \;  {\rm{Tr}} \left[ {  {\mathbf{R}}_{x} {\dot{{\mathbf
                    A}}}_{2,{\rho_s}}^{\rm{H}} {{\dot{{\mathbf
                    A}}}_{2,{\rho_s}}} } \right] \mathop  = \limits^{(b)} \frac{{4{\pi ^2}{P_{\max }}{N_r} {\Phi}_{{\rho}_{s}} {\rm{(}}{\phi _s}{\rm{,}}{\rho _s}{\rm{,}}{{\rm{y}}_s}{\rm{)}}}}{{\lambda _c^2}},
                \end{equation}
                \begin{equation}
                    {\rm{Tr}} \left[ {  { {\mathbf{R}}_{x} } {\dot{{\mathbf
                    A}}}_{2,y_s}^{\rm{H}} {{\dot{{\mathbf
                    A}}}_{2,y_s}} } \right]=\frac{{8{\pi ^2}y_s^2{P_{\max }}{N_r}}}{{\lambda _c^2}}
                    \left[ {  \frac{1}{{\sqrt {\gamma _2^2 - \gamma _1^2} }} - \frac{4}{{{\pi ^2}({\gamma _2} + {\gamma _1})}}{K^2}\left( {\sqrt {\frac{{2{\gamma _1}}}{{{\gamma _2} + {\gamma _1}}}} } \right) } \right],
                \end{equation}
                 \hrulefill
            \end{subequations}          
            \end{small}
        \end{figure*}

\noindent

\subsection{Proof of {\textbf{Lemma}} ~\ref{Pro:Case2_decompose}}
\renewcommand{\theequation}{D.\arabic{equation}}
\setcounter{equation}{0}
\label{App:Case2_decompose} 
    It has shown that $ {\mathop{\min }} \;{\rm{CRB}}_{{\phi
    }_{s}}$ is equivalent to ${\mathop{\max }} \; {  \left| { {{\mathbf{h}}_{s}^{\rm{H}}}{{\mathbf{w}}} } \right|^2} $ and ${\mathop{\min }} \; {\rm{CRB}}_{ {\rho}_{s} } $ can be approximated by $ {\mathop{\max }} \; {  \left| { {{\tilde{{\bm{\alpha}}}}_{{\rho}_{s}}^{\rm{H}}}{{\mathbf{w}}} } \right|^2}$ still holds for non-coplanar case. Meanwhile, $[ {{\bm{\alpha}}_{t,2}} \odot {{\mathbf{v}}_{t,2,3}} ]^{\rm{H}} {{\bm{\alpha}}_{t,2}} \in {\mathbb{R}}$, which indicates that $ {{\bm{\alpha}}_{t,2}} \odot {{\mathbf{v}}_{t,2,3}} $ is parallel to ${{\bm{\alpha}}_{t,2}}$. In other words, there is an eigenvector ${\bar{{\bm{\alpha}}}}_{2,3} \in {\mathbb{C}}^{{N}_{t} \times 1}$ of ${\dot{{\mathbf{A}}}}_{2,{y_s}}^{\rm{H}} {\dot{{\mathbf{A}}}}_{2,{y_s}} $ which satisfied ${\bar{{\bm{\alpha}}}}_{2,3}^{\rm{H}} {\bm{\alpha}}_{t} \in {\mathbb{R}}$. The expression of ${\bar{{\bm{\alpha}}}}_{2,3} $ is given by $ {\bar{{\bm{\alpha}}}}_{2,3}=  \frac{{ {{\bm{\alpha}} }_{t,2}^{\rm{H}} {\dot{{\mathbf{A}}}}_{{\rho_s}}^{\rm{H}} {\dot{{\mathbf{A}}}}_{{\rho_s}} {{\bm{\alpha}} }_{t,2} } } {{ {\bar{{\bm{\alpha}} }}_{t}^{\rm{H}} {\bm{\alpha}}_{t} }} \odot {\bar{{\bm{\alpha}} }}_{t}$. Thus, ${\mathop{\min }} \;{\rm{CRB}}_{{\rho
    }} $ can be transformed into ${\mathop{\max }} \; { \left| { { {\bar{{\bm{\alpha}}}}_{2,3}^{\rm{H}}} {{\mathbf{w}}} } \right|^2} $. 
		
	\vspace{-0.2cm}
	\bibliographystyle{IEEEtran}
	\bibliography{IEEEabrv,ref_AntennaGeometry}   
\end{document}